\documentclass{lmcs} 
\pdfoutput=1

\usepackage{lastpage}
\lmcsdoi{20}{4}{26}
\lmcsheading{}{\pageref{LastPage}}{}{}%
{Sep.~06,~2023}{Dec.~24,~2024}{}

\usepackage[utf8]{inputenc}

\keywords{Worst-case input generation, resource analysis, session types, concurrent programming, amortized analysis.}

\usepackage{graphicx}
\graphicspath{{./figures/}} 

\usepackage{list-packages}

\newcommand{\uscore}{\mbox{\tt\char`\_}}

\newcommand{\m}[1]{\mathsf{#1}}

\usepackage{shortcuts}

\newcommand{\inputs}[1]{\text{inputs} (#1)}
\newcommand{\outputs}[1]{\text{outputs} (#1)}
\newcommand{\proc}[1]{\m{proc} (#1)}
\newcommand{\msg}[1]{\m{msg} (#1)}

\newcommand{\at}{\mathbin{@}}

\newcommand{\true}{\m{true}}
\newcommand{\false}{\m{false}}

\newcommand{\unittype}{\m{unit}}
\newcommand{\bool}{\m{bool}}
\newcommand{\integer}{\m{int}}

\newcommand{\judgment}[1]{\hfill \fbox{#1}}

\newtheoremstyle{cited}%
{.5\baselineskip\@plus.2\baselineskip
	\@minus.2\baselineskip}
{.5\baselineskip\@plus.2\baselineskip
	\@minus.2\baselineskip}
{\itshape}
{\parindent}
{}
{.}
{.5em}
{\textsc{\thmname{#1}} \thmnote{\normalfont#3}}

\theoremstyle{cited}
\newtheorem{citedthm}[thm]{Theorem}

\theoremstyle{plain} 

\begin{document}

\title[Worst-Case Input Generation for Concurrent Programs]{Worst-Case Input Generation for Concurrent Programs under Non-Monotone Resource Metrics}

\author[L.~Pham]{Long Pham\lmcsorcid{0000-0001-5153-8140}}
\author[J.~Hoffmann]{Jan Hoffmann\lmcsorcid{0000-0001-8326-0788}}

\address{Carnegie Mellon University}	
\email{longp@andrew.cmu.edu, janh@andrew.cmu.edu}  







\begin{abstract}
	Worst-case input generation aims to automatically generate inputs that exhibit
	the worst-case performance of programs.
	It has several applications, and can, for example, detect vulnerabilities to
	denial-of-service (DoS) attacks.
	However, it is non-trivial to generate worst-case inputs for concurrent
	programs, particularly for resources like memory where the peak cost depends
	on how processes are scheduled.

	This article presents the first sound worst-case input generation algorithm
	for concurrent programs under non-monotone resource metrics like memory.
	The key insight is to leverage resource-annotated session types and symbolic
	execution.
	Session types describe communication protocols on channels in process calculi.
	Equipped with resource annotations, resource-annotated session types not only
	encode cost bounds but also indicate how many resources can be reused and
	transferred between processes.
	This information is critical for identifying a worst-case execution path
	during symbolic execution.
	The algorithm is sound: if it returns any input, it is guaranteed to be a
	valid worst-case input.
	The algorithm is also relatively complete: as long as resource-annotated
	session types are sufficiently expressive and the background theory for SMT
	solving is decidable, a worst-case input is guaranteed to be returned.
	A simple case study of a web server's memory usage demonstrates the utility of
	the worst-case input generation algorithm.
\end{abstract}

\maketitle


\section{Introduction}

Understanding the worst-case performance of programs and when it is triggered
helps programmers spot performance bugs and take preemptive measures against
algorithmic complexity attacks.
As pioneered by WISE~\cite{Burnim2009}, symbolic execution is a well-studied
technique for worst-case input generation.
In WISE~\cite{Burnim2009} and SPF-WCA~\cite{Luckow2017}, they first symbolically
execute a program on all inputs of small sizes to identify a worst-case
execution path $p_{\text{short}}$.
This path is then generalized to a longer worst-case execution path
$p_{\text{long}}$ for a larger input.
During the symbolic program along the execution path $p_{\text{long}}$, its
corresponding path constraint is collected.
Finally, solving the path constraint yields an inferred worst-case input for a
large input size.
Since these techniques do not explore the entire search space of large inputs,
they are scalable but unsound.

To achieve soundness in worst-case input generation, Wang and
Hoffmann~\cite{Wang2019} propose type-guided worst-case input generation for
functional programming.
In their algorithm, symbolic execution of a functional program under analysis is
guided by a \emph{resource-annotated type} $\tau_{\text{ra}}$, which is
automatically inferred by the type-based resource analysis technique Automatic
Amortized Resource Analysis (AARA)~\cite{Hofmann2003,Hoffmann2010,Hoffmann2012}.
The resource-annotated type $\tau_{\text{ra}}$ encodes a sound (but not
necessarily tight) polynomial worst-case bound.
To identify a worst-case execution path, Wang and Hoffmann's algorithm searches
for an execution path where the cost bound encoded by the resource-annotated
type $\tau_{\text{ra}}$ is tight.
Solving the path constraint of this worst-case execution path, we obtain a valid
worst-case input: it has the same cost as the cost bound captured by the
resource-annotated type $\tau_{\text{ra}}$, which is not only sound but also
tight.

All existing techniques for worst-case input generation, however, cannot handle
the joint setting of
\begin{enumerate*}[label=(\roman*)]
	\item concurrent programming and
	\item non-monotone resource metrics (e.g., memory).
\end{enumerate*}
A resource metric is \emph{non-monotone} if resources can be freed up as well as
consumed.
Worst-case input generation for this joint setting has a practical value.
For example, denial-of-service (DoS) attacks overwhelm the memory of servers,
which are typically concurrent programs.
Hence, the worst-case input generation for concurrent programs under
non-monotone resource metrics will be able to identify vulnerabilities to DoS
attacks.

This article presents the first sound worst-case input generation algorithm for
message-passing concurrent programming under non-monotone resource metrics.
Our work builds on the type-guided worst-case input generation for functional
programming by Wang and Hoffmann~\cite{Wang2019}.
In extending their algorithm from functional programming to message-passing
concurrent programming, the first challenge is to adapt the notion of
\emph{skeletons}, which specify the shapes and sizes of worst-case inputs to be
generated, to message-passing concurrent programming.
Designing skeletons is complicated by the following unique characteristics of
message-passing concurrent programming:
\begin{itemize}
	\item Interaction between channels: the shapes of inputs on different channels
	      may be dependent on one another.
	      This stands in contrast to functional programming where the shapes of
	      inputs, if there are multiple, are independent of one another.
	\item Co-inductive interpretation: inputs to a concurrent program may be
	      infinite.
	\item Intertwining of input and output: input and output are intertwined,
	      unlike in functional programming where all inputs are provided before
	      program execution.
\end{itemize}

Our first contribution is to design a suitable notion of \emph{session
	skeletons} for message-passing concurrent programming.
Session skeletons are built on \emph{session types}~\cite{Honda1993} that
describe communication protocols on channels in process calculi.

In message-passing concurrent programming, monotone costs like work (i.e.,
sequential running time) are independent of how concurrent processes are
scheduled.
Hence, monotone costs in message-passing concurrent programming can be treated
in the same manner as in Wang and Hoffmann's work for functional programming.
Meanwhile, costs like span (i.e., parallel running time) are dependent on
schedules, but they are outside the scope of this article.
Instead, we are interested in work-like, non-monotone costs such as memory.

Non-monotone resource metrics, wherein resources can be consumed as well as
freed up, have two types of costs: \emph{net cost} (i.e., the net quantity of
resources consumed) and \emph{high-water-mark cost} (i.e., the peak net cost
that has been reached).
For example, suppose 3 units of resources (e.g., memory cells) are consumed at
the beginning of computation, but later, 2 units are freed up at the end of the
computation.
In this case, the net cost is $1 = 3 - 2$, while the high-water mark cost is 3
because, at any point during the computation, the maximum net cost is 3.
In message-passing concurrent programming, a worst-case input is defined as an
input with the maximum high-water-mark cost.

The second challenge in worst-case input generation for message-passing
concurrent programming is to identify a worst-case schedule of concurrent
processes, which is crucial for identifying a worst-case input.
While the net cost of a program is independent of the schedules of concurrent
processes, the high-water-mark cost is dependent on the schedules.
Moreover, the operational semantics of concurrent programming languages do not
specify the exact scheduling of processes.
Consequently, haphazardly performing symbolic execution of a concurrent
program does not necessarily reveal the correct high-water-mark cost of a given
execution path.
We must additionally identify a worst-case schedule of concurrent processes.
However, it is non-trivial to learn such a schedule on the fly during symbolic
execution, unless we preprocess the program beforehand.
Thus, Wang and Hoffmann's algorithm cannot directly be extended to
message-passing concurrent programming under non-monotone resource metrics.

To handle the dependency of high-water-mark costs on schedules, we leverage
\emph{resource-annotated session types}~\cite{Das2018}, whose resource
annotations capture
\begin{enumerate*}[label=(\roman*)]
	\item sound bounds on high-water-mark costs and
	\item how many resources can be reused and transferred between processes.
\end{enumerate*}
Thanks to the availability of sound high-water-mark cost bounds in
resource-annotated session types, like Wang and Hoffmann's algorithm, our
worst-case input generation algorithm is sound.
Furthermore, the information about resource transfer enables us to correctly
track high-water-mark costs of execution paths during symbolic execution.

To summarize, in this article, we make the following contributions:
\begin{itemize}
	\item We present the first sound worst-case input generation algorithm for
	      message-passing concurrent programming under non-monotone resource
	      metrics (e.g., memory).
	\item We propose a suitable notion of skeletons. We also address several
	      technical challenges posed by session types in the design of skeletons.
	\item We prove the soundness (i.e., if the algorithm returns anything, it is a
	      valid worst-case input) and relative completeness (i.e., if AARA is
	      sufficiently expressive and the background theory for SMT solving is
	      decidable, a worst-case input is guaranteed to be returned).
	\item We present a case study of worst-case input generation for a web
	      server's memory usage.
\end{itemize}

The article is structured as follows.
\Cref{sec:overview} provides an overview, describing
\begin{enumerate*}[label=(\roman*)]
	\item the challenge of generating worst-case inputs for concurrent programs
	      under non-monotone resource metrics and
	\item how we overcome this challenge.
\end{enumerate*}
\Cref{sec:resource-aware SILL} presents resource-aware SILL, a message-passing
concurrent programming language equipped with resource-annotated session types.
\Cref{sec:session skeletons} defines skeletons and describes the challenges in
their design.
\Cref{sec:worst-case input generation} presents a worst-case input generation
algorithm guided by resource-annotated session types.
\Cref{sec:case study} demonstrates the algorithm through a case study of a web
server.
Finally, \cref{sec:related work} discusses related work, and
\cref{sec:conclusion} concludes the article.



\section{Overview}
\label{sec:overview}

\paragraph{Processes and channels}

This work uses the message-passing concurrent programming language
SILL~\cite{Caires2010,Toninho2013,Pfenning2015}.
Suppose we are given a process $P$ with two channels $c_1$ and $c_2$.
Communication on the channels $c_1$ and $c_2$ can be bidirectional.
The process $P$ uses the channel $c_1$ as a client and provides the channel
$c_2$ as a provider.
\Cref{fig:topology of SILL programs} (a) depicts the process $P$.
The environment that the process $P$ interacts with is called \emph{the external
	world}.
\Cref{fig:topology of SILL programs} (b) depicts a general SILL program
consisting of multiple concurrent processes.
The network of concurrent processes in SILL must be tree-shaped: each channel is
used by exactly one client, instead of being shared by multiple
clients\footnote{SILL's type system was designed to guarantee deadlock freedom.
	Because cyclic dependency among channels may cause a deadlock, SILL disallows
	channels from being shared.
	This restriction results in a tree-shaped network of concurrent processes.
	Although some cycles of channels are benign, SILL's type system is not
	sophisticated enough to handle them.
	\cite{Balzer2017,Balzer2019} present more fine-grained type systems that use
	the acquire-release primitives to achieve deadlock freedom while permitting
	the sharing of channels.}.

\begin{figure}[t]
	\centering
	\includegraphics[width=0.9\columnwidth]{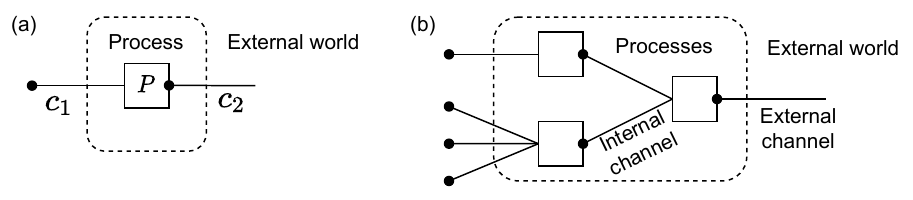}
	\caption{(a) Process $P$ uses channel $c_1$ as a client and provides channel
		$c_2$.
		A dot on a channel denotes the channel's provider.
		(b) General SILL program consisting of multiple concurrent processes
		An internal channel connects two processes; an external channel connects a
		process with the external world.}
	\label{fig:topology of SILL programs}
\end{figure}

We use the following scenario as a running example.
On the channel $c_1$, IP addresses of packets are sequentially sent from the
external world to the process $P$.
The stream of IP addresses may be finite or infinite.
Given the input stream on the channel $c_1$, the process $P$ counts occurrences
of each IP address.
Once the input stream on the channel $c_1$ terminates, $P$ outputs on the
channel $c_2$ the number of IP addresses with at least four occurrences.

Suppose the process $P$ is implemented as follows.
The process maintains a key-value store where keys are IP addresses and values
are the numbers of occurrences.
When a new IP address is encountered, it is added as a key to the key-value
store, incurring the memory cost of 2.
It is because we need one memory cell for the key and another for the value.
When the input stream on the channel $c_1$ terminates, all memory in the
key-value store is released.

\paragraph{Session types}

Session types~\cite{Honda1993} describe communication protocols on channels.
The communication on the channel $c_1$ is described by the session type
\begin{equation}
	\label{eq:session type of input channels in the toy example in the overview}
	\mu X. \oplus \{\m{cons}: \textcolor{red}{\triangleright^{2}} \integer \land X, \m{nil}: \mathbf{1} \}.
\end{equation}
Here, $\mu X$ denotes a recursive session type with a type variable $X$.
The type constructor $\oplus$ means an internal choice, that is, the channel
$c_1$'s provider (i.e., the external world) chooses between labels $\m{cons}$
and $\m{nil}$ and sends the choice.
If the label $\m{cons}$ is chosen, a value of type $\integer$ is sent by the
external world, and we recurse back to $X$.
Conversely, if the label $\m{nil}$ is chosen, the channel $c_1$ is closed.
The resource annotation $\textcolor{red}{\triangleright^{2}}$ will be explained
shortly.

The session type of the channel $c_2$ is $\integer \land \mathbf{1}$.
It means the provider of the channel $c_2$ sends an integer and then closes the
channel by sending the $\m{end}$ message.

\paragraph{Resource annotations}

Resource-aware SILL~\cite{Das2018} incorporates resource annotations into
session types, resulting in resource-annotated session types.
These resource annotations indicate the amount of \emph{potential} necessary to
pay for the computational cost, where the idea of potential comes from the
potential method of amortized analysis for algorithms and data
structures~\cite{Tarjan1985}.
Although potential and resources are similar to each other and hence can used
interchangeably, they are subtly different: potential is an \emph{abstract}
resource used in the potential method, while resources typically refer to
\emph{concrete} computational resources such as time and memory.
In our example, when a previously unseen IP address is encountered, two memory
cells are allocated.
Therefore, in the worst case, we need 2 units of potential to process each IP
address on the channel $c_1$.
This is why the resource-annotated session type of $c_1$ in \cref{eq:session
	type of input channels in the toy example in the overview} contains
$\textcolor{red}{\triangleright^{2}}$.
It denotes that 2 units of potential are transferred from the channel client
(i.e., the external world) to the channel provider (i.e., the process $P$).

Resource-annotated session types are inferred automatically.
Because all numerical constraints generated during type inference are linear,
they can be solved by an off-the-shelf linear-program (LP)
solver~\cite{Das2019}.
Furthermore, the type inference is sound: the cost bounds represented by the
inferred resource-annotated session types are guaranteed to be valid upper
bounds on high-water-mark costs.
Resource-annotated session types in resource-aware SILL can only encode linear
cost bounds, but not polynomial ones\footnotemark.
\footnotetext{While resource-aware SILL~\cite{Das2018} only supports linear
	bounds, Automatic Amortized Resource Analysis
	(AARA)~\cite{Hofmann2003,Hoffmann2010,Hoffmann2012}, which is an analogous
	type-based resource analysis method for functional programs, can infer
	multivariate polynomial bounds.
	Furthermore, all numerical constraints generated during the type inference in
	AARA are linear, even though it can infer multivariate polynomial cost
	bounds.}

\paragraph{Session skeletons}

A SILL program is a network of processes that interact with the external world.
Therefore, an input to a SILL program is a collection of incoming messages from
the external world.
The incoming messages may be intertwined with outgoing messages produced by the
program.

We use the high-water-mark cost, instead of net costs, to define worst-case
inputs.
This definition of worst-case inputs for non-monotone resource metrics (e.g.,
memory), where resources can be freed up as well as consumed, subsumes the
definition of worst-case inputs for monotone resource metrics (e.g., running
time).
In monotone resource metrics, the net cost monotonically increases (without ever
decreasing).
Hence, in monotone resource metrics, the high-water-mark cost (i.e., the maximum
net cost ever reached) is always equal to the net cost.

The first step in worst-case input generation is to provide a \emph{skeleton}
for each external channel.
A skeleton is a symbolic input containing variables, whose concrete values are
to be determined later.
Skeletons specify the shape of worst-case inputs to be generated.

For the channel $c_1$ in the example, a possible skeleton is
\begin{equation}
	\label{eq:skeleton of the input channel in the toy example in the overview}
	\oplus \{\m{cons} : \textcolor{red}{\triangleright^{2}} x_{1} \land \cdots \oplus \{ \m{cons} : \textcolor{red}{\triangleright^{2}} x_{10} \land \oplus \{\m{end}: \mathbf{1} \} \}\cdots \},
\end{equation}
where $x_1, \ldots, x_{10} \in \bbZ$ are integer-typed variables.
This skeleton specifies that the external world should send ten $\m{cons}$'s,
followed by the label $\m{nil}$.
As the channel $c_2$ does not take in any input from the external world, the
channel $c_2$ does not need a skeleton.

Skeletons must satisfy the following requirements:
\begin{itemize}
	\item A skeleton must be compatible with its associated session type. This
	      compatibility relation coincides with the subtyping
	      relation~\cite{Gay2005}.
	\item The input portion of a skeleton must be finite.
\end{itemize}

Because the process $P$ allocates two memory cells whenever a new IP address is
encountered, a worst-case input should have mutually distinct IP addresses on
the channel $c_1$'s input stream.
An example worst-case input that conforms to the skeleton
\labelcref{eq:skeleton of the input channel in the toy example in the overview}
is
\begin{equation}
	\label{eq:example worst-case input of the toy example in the overview}
	\forall 1 \leq i \leq 10. x_{i} = i.
\end{equation}

\paragraph{Symbolic execution}

The goal of worst-case input generation is find an input whose high-water-mark
cost achieves the cost bound inferred by resource-aware SILL.
As resource-aware SILL guarantees the soundness of inferred cost bounds, once we
find an input that has the same cost as an inferred cost bound, the input
immediately qualifies as a worst-case input.

To find a worst-case input, we symbolically execute a program on a skeleton,
searching for an execution path where the program's potential reaches zero since
the last time potential was supplied by the external world.
This strategy correctly identifies a worst-case input.
Suppose, for simplicity, that all necessary potential is supplied at the start
of execution (\cref{fig:tight cost bounds} (a)).
Then the cost bound for memory is tight if and only if the potential reaches
zero at some point.
If the potential never reaches zero, we can lower the cost bound while covering
all computational cost (and without plunging into negative potential), implying
that the cost bound is not tight.

\begin{figure}[t]
	\centering
	\includegraphics[width=0.9\columnwidth]{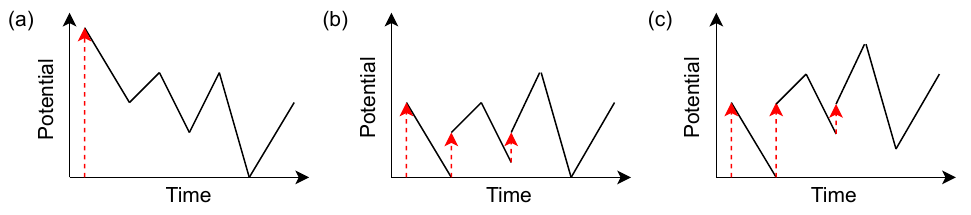}
	\caption{(a) Tight cost bound when all potential is supplied at once.
		The red dashed arrow indicates the potential supplied by the external world.
		(b) Tight cost bound when potential is supplied gradually.
		(c) Loose cost bound.
		Potential zero is never reached after the last injection of potential.
		So we can lower the cost bound (i.e., shorten the second and third red
		dashed arrows) without plunging into negative potential.}
	\label{fig:tight cost bounds}
\end{figure}

In SILL, because potential is supplied to a program gradually (rather than all
at once), we must ensure that the program will eventually reach potential zero
whenever the external world supplies potential to the program (\cref{fig:tight
	cost bounds} (b)).
Otherwise, we could lower the cost bound (\cref{fig:tight cost bounds} (c))
without plunging into negative potential.

\paragraph{High-water-mark costs under concurrency}

In the presence of multiple concurrent processes, their concurrency poses a
challenge: different schedules for symbolic execution may result in different
high-water marks of non-monotone resources.
Generally, in concurrent programming, monotone resources also have dependency on
schedules.
An example is a race condition where two processes compete for a single message
and their monotone cost depends on which process wins.
However, in session-typed concurrent programming like SILL, session types make
the communication more rigid.
As a result, the above example never arises in SILL, making monotone resources
independent of schedules.

To illustrate the dependency of non-monotone resources on schedules, consider
two processes, $P_1$ and $P_2$.
Initially, these two processes run independently.
The process $P_1$ has the high-water mark $h = 4$ and net cost $w = 0$.
Also, the process $P_2$ has $(h, w) = (1, 0)$, i.e., it has the high-water mark
$h = 1$ and the next cost $w = 0$.
Next, the process $P_1$ sends the message to the process $P_2$, thereby
synchronizing them.
We assume that the communication between the processes $P_1$ and $P_2$ is
\emph{asynchronous}.
That is, once it sends a message, the process $P_1$ does not need to wait for
the process $P_2$ to receive the message.
After sending the message, the process $P_1$ incurs $(h, w) = (2, 0)$.
After receiving the message, the process $P_2$ incurs $(h, w) = (3, 0)$.
This situation is depicted in \cref{fig:high-water-mark cost in the presence of
	concurrency} (a).
In the figure, the annotation $\m{tick} \; q$, where $q \in \bbQ$ is a rational
number, means $q$ units of resources are consumed.
As a special case, if $q < 0$, then the annotation $\m{tick} \; q$ means
$\abs{q}$ units of resources are freed up.
The arrows indicate the happened-before relation~\cite{Lamport1978}.

\begin{figure}[t]
	\centering
	\includegraphics[width=0.85\columnwidth]{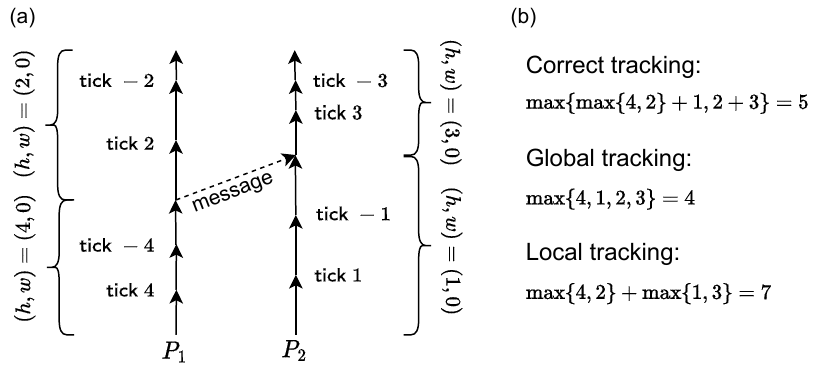}
	\caption{(a) Concurrent processes $P_1$ and $P_2$.
		Time passes by in the direction of arrows.
		(b) Results of global tracking and local tracking.}
	\label{fig:high-water-mark cost in the presence of concurrency}
\end{figure}

The worst-case combined high-water-mark cost of the processes $P_1$ and $P_2$ is
5, and it can be derived as follows.
Firstly, before the process $P_2$ receives the message, their worst-case combined
high-water mark is $\max \{4, 2\} + 1 = 5$.
Due to the asynchrony of communication, by the time the process $P_2$ receives
the message, $P_1$ may have just finished the first phase where $(h, w) = (4,
	0)$ or may already be in the second phase where $(h, w) = (2, 0)$.
Therefore, the high-water mark of $P_1$ before $P_2$ receives the message is
given by $\max \{4, 2\}$.
The high-water mark of $P_2$ before it receives the message is $h = 1$.
If the peak net costs of the processes $P_1$ and $P_2$ happen at the same time,
their worst-case combined high-water mark is $\max \{4, 2\} + 1 = 5$.
Secondly, after the process $P_2$ receives the message, the combined high-water
mark is $2 + 3 = 5$.
This is because, in the worst case, the high-water marks of the processes $P_1$
and $P_2$ in their second phases, namely 2 and 3, happen at the same, resulting
in the combined high-water mark of $2 + 3 = 5$.
Therefore, overall, the worst-case combined high-water mark throughout the
execution is $\max \{5, 5\} = 5$.

\paragraph{Global and local tracking of high-water marks}

To identify a worst-case execution path, we must calculate the tight worst-case
high-water mark of any execution path.
Two naive ways to track costs are global and local tracking.
In global tracking, we have a global cost counter shared by all processes.
In local tracking, each process tracks its own cost.
The combined cost is the sum of all local costs after the program terminates.

Neither global tracking nor local tracking returns the correct worst-case
high-water mark (\cref{fig:high-water-mark cost in the presence of concurrency}
(b)).
On the one hand, global tracking may underestimate it.
Before synchronization, if the peak costs of the processes $P_1$ and $P_2$
happen at different times, the global counter registers a high-water mark below
5.
On the other hand, local tracking may overestimate the worst-case combined
high-water mark.
When the program terminates, the process $P_1$'s local counter registers $\max
	\{4, 2\} = 4$, and $P_2$'s counter registers $\max \{1, 3\} = 3$.
Their sum is 7, which overestimates the tight worst-case combined high-water
mark of 5.

The root cause of overestimation in local tracking is that it does not account
for the possibility that, due to message-passing synchronization, the high-water
marks of two concurrent processes cannot happen at the same time.
For example, in \cref{fig:high-water-mark cost in the presence of concurrency}
(a), the process $P_1$'s local cost counter registers $(h, w) = (4, 0)$
\emph{before} $P_1$ sends the message (and hence also \emph{before} the
process $P_2$ receives the message).
Meanwhile, the process $P_2$'s local cost counter registers $(h, w) = (3, 0)$
\emph{after} the process $P_2$ receives the message.
Thus, the high-water marks 4 (of the process $P_1$) and 3 (of the process $P_2$)
cannot happen at the same time---these two high-water marks must respectively
happen before and after the synchronization of the two processes.
Nonetheless, because local tracking does not account for this fact, it naively
sums the respective high-water marks of the two local counters, resulting in an
overestimate of $3 + 4 = 7$ for the worst-case combined high-water mark.

\paragraph{Transfer of potential}

To resolve this issue of local tracking, our key insight is to \emph{transfer
	potential between processes}.
Suppose the processes $P_1$ and $P_2$ initially store 4 units and 1 unit of
potential, respectively.
In the process $P_1$, after it runs $\m{tick} \; 4$, all potential stored in the
process $P_1$ is consumed.
But the expression $\m{tick} \; {-4}$ in the process $P_1$ frees up 4 units of
potential.
Likewise, the process $P_2$'s potential is used up by the expression $\m{tick}
	\; 1$, but is later freed up by the expression $\m{tick} \; {-1}$.
When the process $P_1$ sends the message, the message also carries 2 units of
potential.
This potential is paid by the process $P_1$, and in turn, it is used to pay for
the process $P_2$'s cost.
\Cref{fig:tight cost bound of the example with the help of potential} (a)
depicts this situation.

\begin{figure}[t]
	\centering
	\includegraphics[width=0.83\columnwidth]{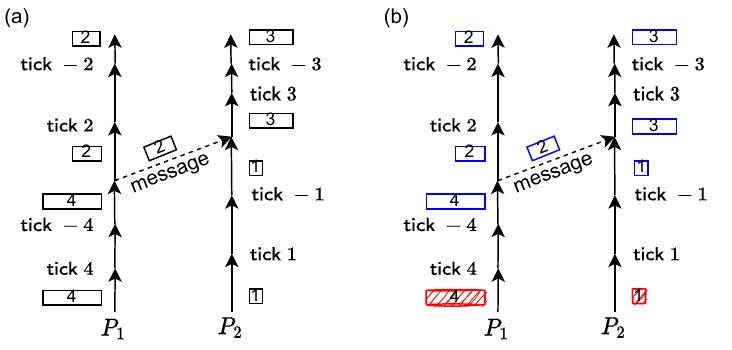}
	\caption{(a) The potential helps us derive the correct high-water mark.
		A rectangle next to an arrow represents the potential during the arrow's
		time period.
		The number inside a rectangle and its length indicate the amount of
		potential.
		If no rectangle exists, it indicates zero potential.
		(b) Same diagram as (a), but the potential is colored red and blue.}
	\label{fig:tight cost bound of the example with the help of potential}
\end{figure}

The combined high-water mark of the processes $P_1$ and $P_2$ is bounded above
by the total potential supplied at the beginning, namely $4 + 1 = 5$, which is
indeed a tight bound.

What makes the potential method more powerful than local tracking is the ability
to transfer potential between processes, allowing the potential to be reused
(and hence shared) by multiple processes.
By reusing and sharing the potential between processes, the local cost counters
of concurrent processes are calibrated such that their sum yields a tighter
combined high-water mark.
To illustrate it, recall that, in \cref{fig:high-water-mark cost in the presence
	of concurrency}, local tracking overestimates the worst-case high-water mark
because local tracking does not account for the possibility that the high-water
marks of two concurrent processes cannot happen at the same time.
To fix this issue, we have the process $P_1$ send 2 units of potential to the
process $P_2$, along the message, such that this potential can be reused by the
process $P_2$.
After the process $P_2$ receives the message, which carries the 2 units of
potential, we use this potential to partially pay for the high-water mark of the
process $P_2$ after the synchronization.
This partial payment allows us to decrease the post-synchronization high-water
mark of the process $P_2$ from 3 to $1 = 3 - 2$.

The key innovation of our worst-case input generation algorithm for concurrent
programs is to leverage \emph{resource-annotated session types}~\cite{Das2018},
which capture information about potential transfer, to guide symbolic execution.
We first automatically infer the resource-annotated session $A$ of a given
concurrent program expressed in SILL.
We then run symbolic execution of the program to systematically and exhaustively
search for a finite execution path (according to user-specified session skeleton
like \cref{eq:skeleton of the input channel in the toy example in the overview})
where the bound on the high-water mark cost encoded by the resource-annotated
session type $A$ is tight.
In addition to the high-water-mark bound, the type $A$ encodes information about
how resources are transferred between processes.
Hence, by exploiting this information during the symbolic execution, for any
execution path, we can tightly track the worst-case high-water mark (among all
possible schedules of concurrent processes), thereby checking the tightness of
the high-water mark bound encoded inside the type $A$.

\paragraph{Check tightness of cost bounds}

Resource-annotated session types already capture sound bounds on high-water
marks.
In order for these bounds to translate into true high-water marks, it remains to
ascertain that the bounds are tight.
This is achieved by tracking individual units of potential during symbolic
execution.

Let us call
\begin{enumerate*}[label=(\roman*)]
	\item the potential supplied by the external world \emph{\textcolor{red}{red}
		      potential} and
	\item the potential freed up in a process \emph{\textcolor{blue}{blue}
		      potential}.
\end{enumerate*}
If a process executes $\m{tick} \; q$ for $q > 0$ and stores no potential, the
external world must supply at least $q$ units of (\textcolor{red}{red})
potential to the process.
Conversely, if a process executes $\m{tick} \; q$ for $q < 0$, then $\abs{q}$
units of (\textcolor{blue}{blue}) potential are freed up and become available to
the process.
A cost bound of an entire concurrent program is given by the total red potential
supplied by the external world to the program.

A cost bound is tight if it satisfies two conditions.
Firstly, red potential must be consumed completely.
Otherwise, we could lower the cost bound while paying for all computational
costs.
Using the same $P_1$ and $P_2$ from \cref{fig:high-water-mark cost in the
	presence of concurrency}, \cref{fig:blue potential and red potential} (a)
illustrates a situation where red potential is not consumed entirely.
The processes $P_1$ and $P_2$ are initially given a total of $4 + 2 = 6$ units
of red potential.
But 0.5 units of red potential are left unconsumed in the process $P_2$,
suggesting that the cost bound of 6 is not tight.

\begin{figure}[t]
	\centering
	\includegraphics[width=0.83\columnwidth]{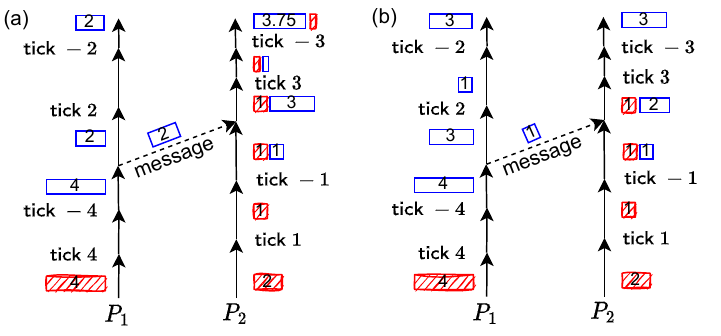}
	\caption{(a) Red potential in the process $P_2$ is not entirely consumed.
		Red hatched rectangles and blue blank rectangles represent red and blue
		potential, respectively.
		(b) Blue potential is generated in the process $P_1$ before red potential is
		consumed in the process $P_2$.
		But some blue potential is left unconsumed in the process $P_1$.}
	\label{fig:blue potential and red potential}
\end{figure}

Secondly, every unit of blue potential must be consumed if its generation
precedes the consumption of red potential.
Assume otherwise: blue potential is not consumed entirely, while red potential
is consumed after that blue potential was generated.
An example is illustrated in \cref{fig:blue potential and red potential} (b).
As in part (a), the processes $P_1$ and $P_2$ are initially given 4 units and
2 units of red potential, respectively.
However, this time, the process $P_1$ only sends 1 unit, instead of 2 units,
of (blue) potential to the process $P_2$.
Consequently, 1 unit of the blue potential generated by $\m{tick} \; {-4}$ on
the process $P_1$ remains unconsumed in the rest of $P_1$'s lifetime.
Furthermore, $\m{tick} \; {-4}$ happens before red potential is consumed by
$\m{tick} \; 3$ on the process $P_2$.
Hence, we could send more blue potential from the process $P_1$ to the process
$P_2$, thereby substituting the blue potential for red potential supplied to the
process $P_2$ by the external world.
This means the cost bound of 6 is, again, not tight.

Finally, \cref{fig:tight cost bound of the example with the help of potential}
(b) illustrates the case where the above two conditions are met.
The cost bound is indeed tight.

\paragraph{Solve path constraints}

Once a worst-case execution path is identified, its path constraint is fed to an
SMT solver to generate a concrete worst-case input.
In our example \cref{eq:session type of input channels in the toy example in the
	overview}, the worst-case execution path is where each incoming IP address
requires 2 units of potential.
Solving the path constraint of this path, we obtain a set of mutually distinct
IP addresses such as \cref{eq:example worst-case input of the toy example in the
	overview}.

Current SMT technologies cannot solve constraints over an infinite data
structure.
Consequently, it is necessary to require the input portion of a session skeleton
to be finite; otherwise, we would not be able to find an infinitely large
worst-case input by solving its path constraint using an SMT solver.
An example of a path constraint over an infinite worst-case input that modern
SMT solvers cannot handle is given in \cref{eq:constraint on function f that
	encodes an infinite stream of Booleans} (\cref{sec:finite input portion of
	skeletons}).


\section{Resource-Aware SILL}
\label{sec:resource-aware SILL}

Resource-aware Simple Intuitionistic Linear Logic (SILL)~\cite{Das2018} has two
constructs: processes and channels.
Processes send and receive messages, including channels, on channels.
A channel in SILL connects two processes: a provider and a client.
They can communicate in both directions\footnote{Even though a provider and a
	client can communicate in both directions, we still assign different roles to
	the two endpoints of a channel due to the correspondence between SILL and
	intuitionistic linear logic~\cite{Caires2010}.}.
While each process must provide exactly one channel, the process can be a client
of multiple (possibly zero) channels.

Channels are typed with (resource-annotated) session types, which describe the
communication protocols on the channels.
Well-typedness of channels in SILL guarantees deadlock freedom (i.e., progress)
and session fidelity (i.e., preservation)~\cite{Caires2010}.

\subsection{Session Types}
\label{sec:session types}

Resource-aware SILL has two layers: functional layer and process layer.
Types $\tau$ in the functional layer and resource-annotated session types $A$ in
the process layer are formed by the following grammar:
\begin{align*}
	b \Coloneq {}    & \unittype \mid \bool \mid \integer \mid b_1 \times b_2 \mid b_1 + b_2 &  & \text{primitive and polynomial types}                                                    \\
	\tau \Coloneq {} & \{c: A \leftarrow \overline{c_i: A_i} \at q \}                        &  & \text{process type}                                                                      \\
	                 & \mid b \mid \tau_1 \to \tau_2                                         &  & \text{functional types}                                                                  \\
	A \Coloneq {}    & X \mid b \supset A \mid \cdots \mid \triangleright^{q} A              &  & \text{resource-annotated session types (\cref{table:session types and process terms})}.
\end{align*}

\paragraph{Functional types}

In the functional layer, if an expression $e$ has a process type \linebreak[5] $\{c: A
	\leftarrow \overline{c_i: A_i} \at q \}$, the expression $e$ represents a
process that
\begin{enumerate*}[label=(\roman*)]
	\item provides a channel $c$ of resource-annotated session type $A$ and
	\item uses channels $c_1, \ldots, c_n$ of resource-annotated session types
	      $A_1, \ldots, A_n$.
\end{enumerate*}
The annotation $q \in \bbQ_{\geq 0}$ denotes the constant potential initially
stored in the process.
The rest of the types in the functional layer are standard.

\paragraph{Session types}


In \cref{table:session types and process terms}, the first column lists all
resource-annotated session types.
The second column lists the continuations of the session types from the first
column, that is, what session type the channel will have after it makes one
action (e.g., send or receive a message).
The last column describes, for each session type $A$, the operational semantics
of a channel $c$ with the session type $A$ from the viewpoint of the the channel
$c$'s provider.

Briefly, the session type $b \supset A$ (and $b \land A$) receives (and sends) a
value of a functional type $b$ and then proceeds to the session type $A$.
The session type $A_1 \multimap A_2$ (and $A_1 \otimes A_2$) receives (and
sends) a channel of session type $A_1$ and then proceeds to the session type
$A_2$.
The session type $\& \{ \overline{\ell_{i}: A_{i}} \}$ (and $\oplus \{
	\overline{\ell_{i}: A_{i}} \}$) receives (and sends) a label $\ell_{j}$ for some
$j$ and proceeds to the session type $A_j$ accordingly.
The session type $\mathbf{1}$ closes the channel and sends the message $\m{end}$
to the channel client in order to signal the closure.
The session type $\triangleleft^{q} A$ (and $\triangleright^{q} A$) receives
(and sends) $q \in \bbQ_{> 0}$ units of potential to the channel client.

Unlike the original SILL~\cite{Toninho2013,Pfenning2015}, resource-annotated
SILL~\cite{Das2018} does not offer the exponential operator $!$ for copying a
channel.
Hence, resource-aware SILL is related, via the Curry-Howard correspondence, to
intuitionistic multiplicative-additive linear logic (IMALL)~\cite{Lincoln1993}.

We assume a global signature $\Sigma$ containing definitions of type variables
of the form $X = A_{X}$, where $X$ is a type variable and $A_{X}$ is a session
type that may mention $X$ (hence recursive).
Recursive session types are interpreted co-inductively, so communication can
last forever.
Also, we require recursive session types to be contractive~\cite{Gay2005}, that
is, we cannot have a recursive session type $\mu X. X$ that remains identical
after unfolding the recursive operator.
Lastly, recursive session types are regarded equi-recursive.
Hence, throughout this article, type variables can be silently replaced by their
definitions.

The type inference algorithm attempts to automatically determine resource
annotations $q \in \bbQ_{> 0}$ in $\triangleright^{q}$ and
$\triangleleft^{q}$~\cite{Das2019} by collecting linear constraints from
according to the type system in \cref{sec:type system of resource-aware SILL}
and solving them by a linear-program (LP) solver.
So a user does not need to manually provide the values of $q$.


\begin{table}[t]
  \centering
  \caption{Resource-annotated session types and process terms.
    In the last column, their operational semantics are described from the
    viewpoint of channel providers.
    For the session type $\mathbf{1}$, the first row is for the provider of the
    channel, and the second row is for the client of the channel.}
  \label{table:session types and process terms}
  \begin{adjustbox}{max width=\textwidth}
    \begin{tabular}{l l l l l}
      \toprule
      Session type                              & Cont.   & Process term                                                      & Cont.   & Operational semantics                                        \\
      \midrule
      $b \supset A$                             & $A$     & $x \leftarrow \m{recv} \; c; P_{x}$                               & $P_{x}$ & receive a value of base type $b$ on channel $c$              \\
      ~                                         & ~       & ~                                                                 & ~       & and bind it to variable $x$                                  \\
      $b \land A$                               & $A$     & $\m{send} \; c \; v; P$                                           & $P$     & send a value $v: b$ on channel $c$                           \\
      $A_1 \multimap A_2$                       & $A_2$   & $x \leftarrow \m{recv} \; c; P_{x}$                               & $P_{x}$ & receive a channel of type $A_1$ on channel $c$               \\
      ~                                         & ~       & ~                                                                 & ~       & and bind it to variable $x$                                  \\
      $A_1 \otimes A_2$                         & $A_2$   & $\m{send} \; c \; d; P$                                           & $P$     & send a channel $d: A_1$ on channel $c$                       \\
      $\& \{ \overline{\ell_{i}: A_{i}} \}$     & $A_{j}$ & $\m{case} \; c \; \{\overline{\ell_{i} \hookrightarrow P_{i}} \}$ & $P_{j}$ & receive a label on $c$ and conduct pattern matching          \\
      $\oplus \{ \overline{\ell_{i}: A_{i}} \}$ & $A_{j}$ & $c.\ell_{j}; P$                                                   & $P$     & send a label $\ell_{j}$ on channel $c$                       \\
      $\mathbf{1}$                              & N/A     & $\m{close} \; c$                                                  & N/A     & close channel $c$ by sending the $\m{end}$ message           \\
      ~                                         & ~       & $\m{wait} \; c; P$                                                & $P$     & wait for the $\m{end}$ message of channel $c$'s closure      \\
      $\triangleleft^{q} A$                     & $A$     & $\m{get} \; c \; \{q\}; P$                                        & $P$     & receive $q \in \bbQ_{> 0}$ units of potential on channel $c$ \\
      $\triangleright^{q} A$                    & $A$     & $\m{pay} \; c \; \{q\}; P$                                        & $P$     & send $q \in \bbQ_{> 0}$ units of potential on channel $c$    \\
      \bottomrule
    \end{tabular}
  \end{adjustbox}
\end{table}

\subsection{Syntax}
\label{sec:syntax of resource-aware SILL}

Fix a set $\calF$ of function identifiers.
A program in resource-aware SILL is a pair $(P, \Sigma)$, where $P$ is the main
process to run and the signature $\Sigma$ stores type definitions and function
definitions.
The syntax of functional terms $e$ and processes $P$ is given below.
\begin{align*}
	e \Coloneq {} & x \mid \tuple{\,} \mid \true \mid \false \mid n \in \bbZ \mid f \in \calF \mid \cdots &  & \text{standard functional terms}                                      \\
	              & \mid c \leftarrow P_{c, \overline{c_i}} \leftarrow \overline{c_i}                     &  & \text{process constructor}                                            \\
	P \Coloneq {} & c \leftarrow e \leftarrow \overline{c_i}; P_{c}                                       &  & \text{spawn a process}                                                \\
	              & \mid c_1 \leftarrow c_2                                                               &  & \text{forward messages}                                               \\
	              & \mid \m{tick} \; q; P                                                                 &  & \text{consume resources}                                              \\
	              & \mid x \leftarrow \m{recv} \; c; P_{x} \mid \cdots \mid \m{pay} \; c \; \{q\}; P      &  & \text{process terms (\cref{table:session types and process terms})}.
\end{align*}

\paragraph{Functional terms}

The term $c \leftarrow P_{c, \overline{c_i}} \leftarrow \overline{c_i}$
encapsulates a process $P_{c, \overline{c_i}}$ that provides a channel $c$ and
uses channels $\overline{c_i}$.
The rest of the functional layer's syntax is standard.

\paragraph{Processes}

In the process layer, the process $c \leftarrow e \leftarrow \overline{c_i};
	P_{c}$ first spawns a new process denoted by $e$.
This child process provides a channel $c$ and uses channels $\overline{c_i}$ as
a client.
After spawning $e$, the parent process continues as a process $P_{c}$ and uses
the channel $c$ as a client.
The process $c_1 \leftarrow c_2$ forwards messages between channels $c_1$ and
$c_2$ in both directions.
The process $\m{tick} \; q; P$ consumes $q \in \bbQ$ units of resources.
The construct $\m{tick} \; q$ is inserted either manually by a user or
automatically according to a resource metric of the user's interest.
For instance, if a user is interested in memory, allocation of a 64-bit integer
is modeled as $\m{tick} \; 64$.
The rest of the syntax is given in the third and fourth columns of
\cref{table:session types and process terms}.
The value of $q$ in processes $\m{get} \; c \; \{q\}; P$ and $\m{pay} \; c \;
	\{q\}; P$ is automatically inferred~\cite{Das2018}.

For illustration, consider a process $P$ that provides a channel $c$ and uses
no channels.
Suppose that the process $P$ receives an integer $x \in \bbZ$ on the channel $c$
and spawns a new process that
\begin{enumerate*}[label=(\roman*)]
	\item provides a new channel $f$ and
	\item sends an integer $x+1$.
\end{enumerate*}
Lastly, the process $P$ closes the channel $c$.
An implementation of the process $P$ is
\begin{align}
	P     & \coloneq x \leftarrow \m{recv} \; c; d \leftarrow f \, x; \m{close} \; c \label{eq:example process definition of P} \\
	f (x) & = \m{send} \; c \; (1 + x); \m{close} \; c. \label{eq:definition of function f that takes an integer x as input and returns a process}
\end{align}
The definition \labelcref{eq:example process definition of P} of the process $P$
calls a function $f: \integer \to \{c \leftarrow \cdot \at 0\}$, whose type
signature means that the function $f$ takes an integer as input and returns a
process of type $\{c \leftarrow \cdot \at 0\}$, i.e., it provides a channel $c$,
uses no channels, and requires zero initial potential.
When the process $P$ runs $d \leftarrow f \, x$, the channel $d$ is substituted
for the channel $c$ used inside the function body \labelcref{eq:definition of
function f that takes an integer x as input and returns a process} of $f$.
%

A resource metric is said to be monotone if every $\m{tick} \; q$ has $q \geq
	0$.
Conversely, if $q < 0$ (i.e., resources are freed up as well as consumed) is
allowed, the resource metric is said to be non-monotone.
Examples of non-monotone resource metrics are memory (e.g., heap space) and
money (e.g., cryptocurrencies transferred by smart contracts).
Non-monotone resource metrics subsume monotone ones.

\subsection{Cost Semantics}
\label{sec:cost semantics}

The cost semantics of SILL is defined using substructural operational
semantics~\cite{Pfenning2009}, which is essentially a multiset rewriting
system~\cite{Cervesato2006}.
A program state is represented by a multiset, called a configuration, of
predicates.
Rewriting rules specify how one configuration transitions to another.
Suppose we are given a rewriting rule
\begin{equation}
	\begin{small}
		\inferrule
		{I_1 \\ I_2 \\ \cdots \\ I_n }
		{J_1 \\ J_2 \\ \cdots \\ J_m}
		\qquad (n, m \in \bbN).
	\end{small}
\end{equation}
If the predicates $I_1, \ldots, I_{n}$ \emph{all} exist in a configuration, the
next configuration is obtained by replacing $I_1, \ldots, I_{n}$ with $J_1,
	\ldots, J_{m}$.
Rewriting rules can be applied in any order.

The cost semantics of SILL uses two predicates: $\proc{c, w, P}$ and $\msg{c, w,
		M}$.
The predicate $\proc{c, w, P}$ represents a process $P$ providing a channel $c$.
The predicate $\msg{c, w, M}$ means a message $M$ is being transferred across a
channel $c$, but has not been received yet.
In the predicate $\proc{c, w, P}$, the net-cost counter $w \in \bbQ$ tracks the
cumulative net cost of the process.
In the predicate $\msg{c, w, M}$, the number $w \in \bbQ$ is used when a process
terminates and transfers the net cost at that point to another process.

\Cref{fig:some rules of the cost semantics of SILL} displays some key rules of
the substructural cost semantics.
The grammar of a message $M$ and the remaining rules of the cost semantics can
be found in \cref{fig:remaining rules of the cost semantics of SILL}
(\cref{appendix:cost semantics and the type system}).

\begin{figure}[t]
	\begin{small}
		\judgment{$\proc{c, w, P}$ and $\msg{c, w, M}$}
		\begin{mathpar}
			\inferrule
			{\proc{d, w, c \leftarrow e \leftarrow \overline{c_i}; Q_{c}} \\ e
				\Downarrow (x \leftarrow P_{x, \overline{x_i}} \leftarrow \overline{x_i})
				\\ c' \text{ is fresh}} {\proc{c', 0, P_{c',\overline{ c_{i}}}} \\
				\proc{d, w, Q_{c'}}}
			\m{spawn}
			\and
			\inferrule
			{\proc{d, w, \m{send} \; c \; e; P} \\ e \Downarrow v \\ c' \text{ is
					fresh}} {\proc{d, w, P [c' / c]} \\ \msg{c', 0, \m{send} \; c \; v; c'
					\leftarrow c}} {\supset} S
			\and
			\inferrule
			{\msg{c', w_1, \m{send} \; c \; v; c' \leftarrow c} \\ \proc{c, w_2, x
					\leftarrow \m{recv} \; c; P_{x}}} {\proc{c, w_1 + w_2, P_{v} [c' /
							c]}} {\supset} R
			\and
			\inferrule
			{\proc{c, w, \m{tick} \; q; P}} {\proc{c, w + q, P}}
			\m{tick}
			\and
			\inferrule
			{\proc{c, w, \m{close} \; c}} {\msg{c, w, \m{close} \; c}} \mathbf{1} S
			\and
			\inferrule
			{\msg{c, w_1, \m{close} \; c} \\ \proc{d, w_2, \m{wait} \; c; P}}
			{\proc{d, w_1 + w_2, P}} \mathbf{1} R
		\end{mathpar}
	\end{small}
	\caption{Key rules of the substructural cost semantics of resource-aware SILL.
		The functional-layer evaluation judgment $e \Downarrow v$ in the rules
		$\m{spawn}$ and ${\supset} S$ means the functional term $e$ evaluates to a
		value $v$.
		The definition of the functional-layer evaluation judgment $e \Downarrow v$
		is standard.
		The remaining rules are available in \cref{fig:remaining rules of the cost
			semantics of SILL} (\cref{appendix:cost semantics and the type system}).}
	\label{fig:some rules of the cost semantics of SILL}
\end{figure}

The rule $\m{spawn}$ considers a configuration where
\begin{enumerate*}[label=(\roman*)]
	\item we have a process $c \leftarrow e \leftarrow \overline{c_i}; Q_{c}$ (in
	      the first premise) and
	\item the functional term $e$ evaluates to a value $x \leftarrow P_{x,
				      \overline{x_i}} \leftarrow \overline{x_i}$ (in the second
	      premise), where $P_{x, \overline{x_i}}$ is a process.
\end{enumerate*}
The two conclusions in the rule $\m{spawn}$ indicate that the next configuration
replaces the premises with the following predicates:
\begin{enumerate*}[label=(\roman*)]
	\item $\proc{c', 0, P_{c',\overline{ c_{i}}}}$ for the newly spawned process
	      where the net-cost counter $w$ is initialized to zero and
	\item $\proc{d, w, Q_{c'}}$ for the parent process.
\end{enumerate*}

The rule ${\supset} S$ concerns a configuration containing a process $\m{send}
	\; c \; e; P$ (in the first premise), which sends a message $e$ on channel
$c$.
To send a message in the operational semantics, the premises in the rule are
replaced with these two predicates:
\begin{enumerate*}[label=(\roman*)]
	\item $\proc{d, w, P [c' / c]}$ for the continuation of the process after
	      sending a message and
	\item $\msg{c', 0, \m{send} \; c \; v; c' \leftarrow c}$ for the message
	      $\m{send} \; c \; v; c' \leftarrow c$ containing the content $v$ that is
	      being transferred across a freshly renamed channel $c'$.
\end{enumerate*}
The message predicate $\msg{c', 0, \m{send} \; c \; v; c' \leftarrow c}$ is then
handled by the rule ${\supset} R$, which concerns a process ready to receive a
message (in the first premise).
Since rewriting rules can be applied in any order, ${\supset} R$ is not
necessarily applied \emph{immediately} after ${\supset} S$.
Hence, communication between processes is asynchronous.

The rule $\m{tick}$ states that, whenever $\m{tick} \; q$ is executed, the
net-cost counter $w$ in the premise $\proc{c, w, \m{tick} \; q; P}$ is
incremented by $q$.
The rule $\mathbf{1} S$ concerns a process $\m{close} \; c$, which seeks to
close the channel $c$ that the process provides (and hence also terminates the
process itself).
To close the channel in the operational semantics, we remove replace this
premise with a predicate $\msg{c, w, \m{close} \; c}$ that carries
\begin{enumerate*}[label=(\roman*)]
	\item a signal for the channel closure to the channel client and
	\item the final net cost $w$ of the process.
\end{enumerate*}
The client then receives the message, together with the final net cost of the
sender, according to the rule $\mathbf{1} R$.

The net cost of a configuration $C$ is the sum of the net-cost counters $w$ in
predicates $\proc{c, w, P}$ and $\msg{c, w, M}$ in the configuration $C$ when
the SILL program terminates (i.e., no more rewrite rules can be applied to the
configuration $C$).
The high-water mark of the configuration $C$ is its maximum net cost that has
ever been reached (i.e., the maximum sum of the net-cost counters $w$ in all
predicates in the configuration $C$) during the execution of the SILL program.

\subsection{Type System}
\label{sec:type system of resource-aware SILL}

Fix a signature $\Sigma$ containing type definitions and function definitions.
A typing judgment of the process layer has the form
\begin{equation}
	\label{eq:typing judgment of resource-annotated session type}
	\Phi; \Delta; q \vdash P :: (c: A).
\end{equation}
The judgment \labelcref{eq:typing judgment of resource-annotated session type}
means the channel $c$ is provided by the process $P$ and has a
resource-annotated session type $A$.
Here, $\Phi$ is a functional-layer typing context, and $\Delta$ is a
process-layer typing context that maps channels to resource-annotated session
types.
The channels in $\Delta$'s domain are used by the process $P$ as a client.
The number $q \in \bbQ_{\geq 0}$ denotes how much potential is initially stored
in the process $P$.

\Cref{fig:some rules of the type system of resource-aware SILL} displays key
rules of resource-aware SILL's type system.
The remaining rules are in \cref{fig:remaining rules of the type system of
	resource-aware SILL} (\cref{appendix:cost semantics and the type system}).

\begin{figure}[t]
	\begin{small}
		\judgment{$\Phi; \Delta; q \vdash P :: (c: A)$}
		\begin{mathpar}
			\inferrule
			{\Phi \vdash e :\{x : A \leftarrow \overline{x_{i}: A_{i}} \at p\}
				\\\Delta_{1} = \{\overline{c_{i}: A_{i}}\} \\ \Phi; \Delta_{2}, c: A; q
				\vdash Q_{c} :: (d: D)} {\Phi; \Delta_1, \Delta_2; p + q \vdash (c
				\leftarrow e \leftarrow \overline{c_{i}}; Q_{c}) :: (d: D)}
			\m{spawn}
			\and
			\inferrule
			{\Phi; \Delta, c: A; p \vdash P :: (d: D) \\ \Phi \vdash e: b} {\Phi;
				\Delta, c: b \supset A; p \vdash \m{send} \; c \; e; P :: (d: D)}
			{\supset} L
			\and
			\inferrule
			{\Phi, x: b; \Delta, p \vdash P_{x} :: (c: A)} {\Phi; \Delta; p \vdash (x
				\leftarrow \m{recv} \; c; P_{x}) :: (c: b \supset A)} {\supset} R
			\and
			\inferrule
			{\Phi; \Delta; q \vdash P :: (c: A) \\ p > q} {\Phi; \Delta; p \vdash P ::
				(c: A)}
			\m{relax}
			\and
			\inferrule
			{\Phi; \Delta; p \vdash P :: (d: D)} {\Phi; \Delta; p+q \vdash \m{tick} \;
				q; P :: (d: D)}
			\m{tick}
		\end{mathpar}
	\end{small}
	\caption{Key rules of the type system of resource-aware SILL.
		The judgment $\Phi \vdash e : \tau$ in the rules $\m{spawn}$ and ${\supset}
			L$ is a functional-layer typing judgment stating that the functional term
		$e$ has a functional type $\tau$.
		The remaining rules are in \cref{fig:remaining rules of the type system of
			resource-aware SILL} (\cref{appendix:cost semantics and the type system}).}
	\label{fig:some rules of the type system of resource-aware SILL}
\end{figure}

The rule $\m{spawn}$ in \cref{fig:some rules of the type system of
	resource-aware SILL} states that a process $c \leftarrow e \leftarrow
	\overline{c_{i}}; Q_{c}$, which spawns a new process-term $e$, is well-typed if
$e$ has a correct process type (in the first premise) and the continuation
process $Q_c$ is well-typed (in the third premise).
Additionally, the initial potential necessary for the process $c \leftarrow e
	\leftarrow \overline{c_{i}}; Q_{c}$ is $p + q$, where $p$ is the initial
potential for the spawned process $e$ and $q$ is the initial potential for the
continuation process $Q_{c}$.

The rule ${\supset} L$ concerns a process $\m{send} \; c \; e; P$, which sends a
value $e$ of the functional type $b$ on the channel $c$ used by the process as a
client.
The dual rule, ${\supset} R$, concerns a process $x \leftarrow \m{recv} \; c;
	P_{x}$, which is ready to receive a message of the functional type $b$ on the
channel $c$ provided by the process.
The rule $\m{relax}$ weakens the resource annotation $q$ in the judgment: if the
initial potential $q$ is sufficient for the process $P$ to run, then any larger
potential $q > p$ works as well.
The rule $\m{tick}$ states that, to run the construct $\m{tick} \; q; P$, we
require the initial potential of $p + q$, where $p$ is the potential necessary
for the continuation process $P$.

\Cref{theorem:soundness of resource-aware SILL} states the soundness of the type
system of resource-aware SILL: given an initial configuration $C$ where the net
cost is zero, the total potential of the configuration $C$ is an upper bound of
the high-water mark cost for any configuration reachable from $C$.

\begin{thmC}[Soundness of resource-aware SILL \cite{Das2018,Das2019}]
	\label{theorem:soundness of resource-aware SILL}
	Given a well-typed configuration initial $C$ whose net cost is zero, let
	$p \in \bbQ$ be the total potential stored in the configuration $C$.
	The total potential in the configuration $C$ is defined as the sum of all $q$
	in judgments $\Phi; \Delta; q \vdash P :: (c: A)$ for all processes $P$ in the
	predicates $\proc{c, \uscore, P} \in C$.
	If $C \to^{\ast} C'$ (i.e., a configuration $C'$ is reachable from the
	configuration $C$) and $h \in \bbQ$ is the high-water-mark cost of the
	configuration $C'$, then $h \leq p$ holds.
	That is, the initial potential $p$ correctly upper-bounds the high-water mark
	cost $h$.
\end{thmC}

To prove \cref{theorem:soundness of resource-aware SILL}, Das et
al.~\cite{Das2018} first define a typing judgment $\Phi; \Delta_1
	\stackrel{E}{\models} C :: \Delta_2$ for a configuration $C$, where $\Phi$ is a
functional-layer typing context and $\Delta_{i}$ ($i = 1,2$) is a process-layer
typing context.
The judgment defines $E \in \bbQ$, called energy, as the sum of the total
potential and the net cost in the configuration $C$.
Das et al. then prove that, as the configuration $C$ evolves, the energy $E$ is
conserved: it remains the same or decreases (due to the weakening of potential).
This is the resource-aware SILL's equivalent of the classic type preservation
theorem, and it implies \cref{theorem:soundness of resource-aware SILL}.
Prior works~\cite{Das2018,Das2019} only prove the soundness under monotone
resource metrics.
Nonetheless, to extend the soundness result to non-monotone resource metrics, it
suffices to re-examine the inductive case for $\m{tick}$.


\newcommand{\countinput}[2]{#1 \; \m{countInput} (#2)}
\newcommand{\countoutput}[2]{#1 \; \m{countOutput} (#2)}
\newcommand{\tpath}[1]{\m{path} (#1)}

\newcommand{\costboundleft}[1]{\lceil \triangleleft \rceil (#1)}
\newcommand{\costboundright}[1]{\lceil \triangleright \rceil (#1)}

\section{Session Skeletons}
\label{sec:session skeletons}

Skeletons are symbolic inputs specifying the shape of worst-case inputs to be
generated.
In the worst-case input generation for functional programming~\cite{Wang2019},
given a function that takes in lists, if we want a worst-case input of length
three, an appropriate skeleton is $[x_1, x_2, x_3]$, where $x_i$ are variables
whose values are to be determined.

Message-passing concurrent programming poses three challenges in the design of
skeletons.
First, in the presence of multiple channels, their skeletons are interdependent
on each other due to the interaction between channels.
Second, inputs to concurrent programs may be infinite.
Third, the input and output are intertwined in such a way that the output
influences the acceptable set of subsequent inputs.
Our design of skeletons works around the first two challenges.
The last challenge, described in \cref{sec:input generation with loose cost
	bounds}, is beyond the scope of this article because, to fully address this
challenge, it is necessary to enrich session types such that they capture more
information, particularly the interdependence between input and output.
The challenge does not affect the relative-completeness theorem of our
worst-case input generation algorithm (\cref{theorem:relative completeness of
	worst-case input generation}), since the theorem assumes the cost bounds, which
we calculate by simply summing all resource annotations in session skeletons,
are tight.

\subsection{Syntax}

Fix a set $\calX_{\text{skeleton}}$ of skeleton variables.
They are placeholders for concrete values in a worst-case input.
Skeletons are formed by the following grammar:
\begin{align*}
	H \Coloneq {} & x \in \calX_{\text{skeleton}} \mid \tuple{\,} \mid \bool \mid \false \mid n \in \bbZ            &  & \text{skeleton variable and constants}                \\
	              & \mid \tuple{H_1, H_2} \mid \ell \cdot H \mid r \cdot H                                          &  & \text{skeleton constructors}                          \\
	K \Coloneq {} & b \supset K \mid H \supset K \mid b \land K \mid H \land K                                      &  & \text{value input/output}                             \\
	              & \mid K_1 \multimap K_2 \mid K_1 \otimes K_2                                                     &  & \text{channel input/output}                           \\
	              & \mid \& \{ \overline{\ell_{i} : K_{i}} \} \mid \&_{x} \{ \overline{\ell_{i} : K_{i}} \}         &  & \text{external choice}; x \in \calX_{\text{skeleton}} \\
	              & \mid \oplus \{ \overline{\ell_{i} : K_{i}} \} \mid \oplus_{x} \{ \overline{\ell_{i} : K_{i}} \} &  & \text{internal choice}; x \in \calX_{\text{skeleton}} \\
	              & \mid X \mid \mathbf{1} \mid \triangleleft^{q} K \mid \triangleright^{q} K.
\end{align*}

The meta-variables $H$ and $K$ stand for, respectively, a skeleton for the
functional layer and a skeleton in the process layer.
The grammar of session skeletons $K$ is similar to that of resource-annotated
session types (\cref{sec:session types}).
One difference is that, in addition to the skeleton $b \supset K$, we have the
skeleton $H \supset K$, where $b$ is a base type and $H$ is a functional-layer
skeleton.
The skeleton $H \supset K$ is used when the input skeleton $H$ is generated by
the external world, whereas the skeleton $b \supset K$ is used when the input
type $b$ is sent by a process.
Likewise, the skeletons $\&_{x} \{ \overline{\ell_{i} : K_{i}} \}$ and
$\oplus_{x} \{ \overline{\ell_{i} : K_{i}} \}$, where the subscripts are
skeleton variables $x \in \calX_{\text{skeleton}}$, are used when the choices
are resolved by the external world.
The subscripts $x \in \calX_{\text{skeleton}}$ records which branch is chosen in
a worst-case input.

\subsection{Compatibility of Skeletons with Session Types}
\label{sec:compatibility of skeletons with session types}

Given an external channel $c: A$ provided by a process, suppose a user provides
a skeleton $K$.
To check the compatibility of the skeleton $K$ with the resource-annotated
session type $A$, we introduce the judgment
\begin{equation}
	\label{eq:judgment of compatibility when the external channel is provided by a process}
	\Phi \vdash K \leqslant A,
\end{equation}
where $\Phi$ is a typing context for functional-layer skeletons.
The judgment \labelcref{eq:judgment of compatibility when the external channel
	is provided by a process} states that the skeleton $K$ is a valid skeleton of
the session type $A$, given that the external channel $c$ is provided by a
process.
\Cref{fig:key rules in the compatibility relation K leqslant A} defines
\cref{eq:judgment of compatibility when the external channel is provided by a
	process}.
Dually, if the external channel is provided by the external world, we use the
dual judgment $\Phi \vdash A \leqslant K$.
Its definition is symmetric to \cref{fig:key rules in the compatibility relation
	K leqslant A}.

\begin{figure}[t]
	\begin{small}
		\judgment{$\Phi \vdash K \leqslant A$}
		\begin{mathpar}
			\inferrule
			{ } {\Phi \vdash \mathbf{1} \leqslant \mathbf{1}}
			\textsc{K:Ter}
			\and
			\inferrule
			{\Phi \vdash H : b \\ \Phi \vdash K \leqslant A} {\Phi \vdash H \supset K
				\leqslant b \supset A}
			\textsc{K:ValIn}
			\and
			\inferrule
			{\Phi; \Delta \vdash K \leqslant A} {\Phi; \Delta \vdash b \land K
				\leqslant b \land A}
			\textsc{K:ValOut}
			\and
			\inferrule [K:ChannelIn]
			{\Phi \vdash A_1 \leqslant K_1 \\ \Phi \vdash K_2 \leqslant A_2} {\Phi
				\vdash K_1 \multimap K_2 \leqslant A_1 \multimap A_2}
			\and
			\inferrule [K:ChannelOut]
			{\Phi; \Delta \vdash K_1 \leqslant A_1 \\ \Phi; \Delta \vdash K_2
				\leqslant A_2} {\Phi; \Delta \vdash K_1 \otimes K_2 \leqslant A_1
				\otimes A_2}
			\and
			\inferrule
			{\forall j \in N'. \Phi \vdash K_{j} \leqslant A_{j} \\ \emptyset \subset
				N' \subseteq N} {\Phi \vdash \&_{x} \{\ell_{i}: K_{i} \mid i \in N'\}
				\leqslant \& \{\ell_{j}: A_{j} \mid j \in N\}}
			\textsc{K:ExtChoice}
			\and
			\inferrule
			{\Phi; \Delta \vdash K \leqslant A} {\Phi; \Delta \vdash \triangleleft^{q}
				K \leqslant \triangleleft^{q} A }
			\textsc{K:Get}
			\and
			\inferrule
			{\forall i \in N. \Phi; \Delta \vdash K_{i} \leqslant A_{i}} {\Phi; \Delta
				\vdash \oplus \{\ell_{i}: K_{i} \mid N\} \leqslant \oplus \{\ell_{i}:
				A_{i} \mid N\}}
			\textsc{K:InChoice}
			\and
			\inferrule
			{\Phi; \Delta \vdash K \leqslant A} {\Phi; \Delta \vdash
				\triangleright^{q} K \leqslant \triangleright^{q} A }
			\textsc{K:Pay}
		\end{mathpar}
	\end{small}
	\caption{Inference rules of the compatibility relation between a session
		skeleton and a resource-annotated session type.
		The judgment $\Phi \vdash H : b$ in \textsc{K:ValIn} means the
		functional-layer skeleton $H$ has a functional type $b$.}
	\label{fig:key rules in the compatibility relation K leqslant A}
\end{figure}

Interestingly, the relations $K \leqslant A$ and $A \leqslant K$ coincide with
the subtyping relation $\leqslant$ of session types~\cite{Gay2005}.
Upon reflection, this makes sense: because the skeleton $K$ admits some of the
semantic objects of the session type $A$, the skeleton $K$ can be considered as
a subtype of the session type $A$.

Due to the rule \textsc{K:Ter}, given a session skeleton $K = \mathbf{1}$, the session type
$\mathbf{1}$ is the only compatible session type.
Hence, a session skeleton $K$ is disallowed from stopping halfway when the
corresponding session type $A$ has not terminated yet, e.g., $\Phi \not\vdash
	\mathbf{1} \leqslant \integer \land \mathbf{1}$.

This restriction eliminates the interdependence between skeletons.
For instance, consider a process $P$ with two channels $c_1$ and $c_2$.
In each iteration, the process $P$ either closes both the channels $c_1$ and
$c_2$ or keeps them open.
If the channels $c_1$ and $c_2$ are open, the process $P$ receives one incoming
message on the channel $c_1$ and two incoming messages on the channel $c_2$.
So if a (worst-case) input on the channel $c_1$ has size $n \in \bbN$, a
(worst-case) input of the channel $c_2$ must have size $2 n$.
That is, there is interdependence between the skeletons of the channels $c_1$
and $c_2$.
Thanks to the rule \textsc{K:Ter}, when a skeleton $K_i$ on a channel $c_i$ ($i
	= 1,2$) terminates, the corresponding session type $A_i$ for the channel $c_i$
must also terminate.
This only happens exactly when (worst-case) inputs on the channels $c_1$ and
$c_2$ are $n$ and $2 n$, respectively, for some $n \in \bbN$.

In the rule \textsc{K:ChannelIn}, the first premise uses the dual judgment.
This is because, in the session type $A_1 \multimap A_2$, the input session type
$A_1$ reverses the roles of the channel provider and client.
In the rule \textsc{K:ExtChoice}, a skeleton $K$ includes all labels from a
non-empty subset $N' \subseteq N$.
As we assume that the channel is provided by a process in the network, the
choice of $i$ in a session type $\& \{\ell_i: A_i \mid i \in N \}$ is made by
the external world.
Therefore, the skeleton is allowed to limit the set of $i$ to choose from.

\subsection{Finite Input Portion of Skeletons}
\label{sec:finite input portion of skeletons}

Worst-case inputs must be finite.
Otherwise, two technical challenges would arise:
\begin{enumerate}
	\item It is non-trivial to define a worst-case input when inputs may be
	      infinite.
	\item Existing SMT solvers cannot solve constraints over infinite worst-case
	      inputs.
\end{enumerate}

\paragraph{Worst-case infinite inputs}

Consider a non-terminating process $\cdot; 0 \vdash P :: (c: A)$,
where\footnote{Although the
	notation $\mu X. A_{X}$ is not officially in the syntax of session types
	(\cref{sec:session types}), we use $\mu X. A_{X}$ to denote an equi-recursive
	session type where $X = \mu X. A_{X}$.}
\begin{equation}
	\label{eq:type in the example illustrating the difficulty of defining worst-case inputs for infinite inputs}
	A \coloneq \& \{ \m{first}: \mu X. \triangleleft^{2} \integer \multimap \integer \otimes X, \m{second}: \mu X. \triangleleft^{1} \integer \multimap \integer \otimes X \}.
\end{equation}
The process $P$ is willing to accept two labels.
If $\m{first}$ is chosen, every iteration on channel $c$ requires 2 units of
potential.
Otherwise, if $\m{second}$ is chosen, every iteration only needs 1 unit of
potential.

It is unclear which scenario should be deemed the worst-case input.
One possible answer is that they both have an equal cost of infinity.
If we spot a recursive session type where every iteration incurs non-zero cost,
then it automatically qualifies as a worst-case input, provided that the path
constraint is solvable.
However, this idea sounds too simplistic.
Another possible answer is that the $\m{first}$ branch results in a higher
cost than $\m{second}$ because the former entails 2 units of cost per
iteration, while the latter incurs only 1 unit of cost per iteration.
Therefore, at any moment in time, the first branch has a higher cumulative cost
than the second branch.
However, this reasoning implicitly treats each iteration in
\cref{eq:type in the example illustrating the difficulty of defining worst-case
	inputs for infinite inputs} equally.
But it is arguable whether the iterations inside the two branches in
\cref{eq:type in the example illustrating the difficulty of defining worst-case
	inputs for infinite inputs} can be treated equally.
For instance, each iteration in the $\m{first}$ branch may take twice as much
time as a single iteration in the $\m{second}$ branch.
However, because SILL provides no information about timing, it is impossible to
tell the cost per unit of time.

\paragraph{Generating infinite data structures}

It is tricky to solve constraint satisfaction problems for infinite data
structures.
By way of example, consider a process $\cdot; 0 \vdash P :: (c: A)$, where $A
	\coloneq \mu X. \bool \multimap X$.
Suppose the process $P$ is implemented such that the only worst-case input is an
alternating sequence of $\true$ and $\false$.
To encode an infinite stream of Booleans, a sensible idea is to use a function
$f: \bbN \to \bool$, where the input is an index in the sequence and the output
is the Boolean value at that index.
The path constraint for the worst-case input, where $\true$ and $\false$
alternate, is
\begin{equation}
	\label{eq:constraint on function f that encodes an infinite stream of Booleans}
	\forall x \in \bbN. f (2 x) = \true \land f (2 x + 1) = \false.
\end{equation}
Here, $f: \bbN \to \bool$ is an uninterpreted function, and we seek a concrete
$f$ that satisfies \cref{eq:constraint on function f
	that encodes an infinite stream of Booleans}.
Unfortunately, the current SMT technologies are incapable of finding suitable
$f$ in this example.
We tested SMT solvers Z3~\cite{deMoura2008} and CVC4~\cite{Barrett2011} on
the SMT-LIB2 encoding of \cref{eq:constraint on function f that encodes an
	infinite stream of Booleans} (\cref{appendix:Checking Finiteness of
	Input Portions}), and neither of them could verify the satisfiability.

To go around these two challenges of infinite worst-case inputs, we require the
input portion of a skeleton to be finite.
\Cref{appendix:Checking Finiteness of Input Portions} provides further details.

\subsection{Input Generation with Loose Cost Bounds}
\label{sec:input generation with loose cost bounds}

In the presence of multiple channels, it is non-trivial to calculate a precise
cost bound due to the intertwining of the input and output across different
channels.
For illustration, consider a process $P$ with two external channels
\begin{equation}
	\label{eq:typing judgment of P in the example illustrating the difficulty of extracting a cost bound}
	c_1: A_1; 0 \vdash P :: (c_2: A_2),
\end{equation}
where
\begin{equation}
	A_1 \coloneq \oplus \{\m{expensive}: \textcolor{red}{\triangleright^{2}}
	\mathbf{1}, \m{cheap}: \mathbf{1} \}
	\qquad
	A_2 \coloneq \oplus
	\{\m{expensive}: \textcolor{red}{\triangleleft^{3}} \mathbf{1}, \m{cheap}:
	\mathbf{1} \}.
\end{equation}
The external world first chooses between $\m{expensive}$, which requires 2
units of potential, and $\m{cheap}$, which requires no potential.
The process $P$ next chooses between $\m{expensive}$ and $\m{cheap}$.
If the process $P$ chooses $\m{expensive}$, 3 units of potential are sent as
input to the process $P$; otherwise, no extra potential is required.

The input and output are intertwined in this example.
The process $P$ first inputs a label (possibly with potential) from the external
world, then outputs a label, and lastly inputs potential again.
Additionally, the input and output happen on different channels: the first input
happens on the channel $c_1$, while the output and second input (i.e., 3 units
of potential) happen on the channel $c_2$.

According to the judgment \labelcref{eq:typing judgment of P in the example
	illustrating the difficulty of extracting a cost bound}, the total cost bound
of the process $P$ seems $2 + 3 = 5$.
It is achieved when $\m{expensive}$ is selected on both the channels $c_1$ and
$c_2$.
Hence, to achieve the worst-case cost, the external world should choose
$\m{expensive}$ on the channel $c_1$.
Hopefully, the process $P$ will choose $\m{expensive}$ as well so that the cost
bound of 5 is fulfilled.

However, the process $P$ may fail to choose $\m{expensive}$ on the channel
$c_2$.
The choice of the process $P$ on the channel $c_2$ may depend on the external
world's choice on the channel $c_1$ in such a way that we cannot have
$\m{expensive}$ on both channels.
For instance, suppose the process $P$ is implemented such that it sends the
opposite label to whatever is received on the channel $c_1$:
\begin{equation}
	\label{eq:code of P in the example illustrating the difficulty of extracting a cost bound}
	P \coloneq \m{case} \; c_1 \; \{ \m{expensive} \hookrightarrow \m{tick} \; 2; c_2.\m{cheap}, \m{cheap} \hookrightarrow \m{tick} \; 3; c_2.\m{expensive} \}.
\end{equation}
The tight cost bound is 3, which is lower than the bound deduced from
\cref{eq:typing judgment of P in the example illustrating the difficulty of
	extracting a cost bound}.
As a result, the worst-case input generation algorithm fails because it cannot
find an execution path where the cost bound is tight.

In fact, SILL is already expressive enough to derive the tight cost bound of 3
for the example \labelcref{eq:code of P in the example illustrating the
	difficulty of extracting a cost bound}.
This is evidenced by another valid typing judgment of the process $P$:
\begin{equation}
	\label{eq:second typing judgment of P in the example illustrating the difficulty of extracting a cost bound}
	c_1: A; 3 \vdash P :: (c_2: A),
\end{equation}
where $A \coloneq \oplus \{\m{expensive}: \mathbf{1}, \m{cheap}: \mathbf{1} \}$.
In \eqref{eq:second typing judgment of P in the example illustrating the
	difficulty of extracting a cost bound}, all necessary potential comes from the
initial constant potential 3 stored in the process $P$, which is tighter than
the cost bound of 5.
Thus, typing judgments can misrepresent cost bounds, as exemplified by
\cref{eq:typing judgment of P in the example illustrating the difficulty of
	extracting a cost bound}, even when resource-aware SILL is capable of deriving
tight cost bounds.

The root cause is that session types are not rich enough to capture information
about the interdependence between input and output.
Suppose resource annotations are scattered over a session type.
When inputs, including the supply of incoming potential, are interspersed with
outputs, an early input affects an output, which in turn affects a later input's
resource-annotated session type.
Consequently, some combinations of inputs may be infeasible.
Furthermore, if the input and output reside on different channels,
resource-annotated session types do not tell us how the paths on different
channels are linked with each other.
In the above example, if the type system returns the typing judgment
\labelcref{eq:second typing judgment of P in the example illustrating the
	difficulty of extracting a cost bound}, we obtain a precise cost bound of 3.
However, if the type system returns the typing judgment \labelcref{eq:typing
	judgment of P in the example illustrating the difficulty of extracting a cost
	bound}, which is an equally valid typing judgment, we obtain a loose cost bound
of 5.
Thus, even if SILL's type system can figure out the interdependence between
input and output on different channels, this information is not captured by
session types.
Therefore, to calculate a precise cost bound, it is not sufficient to just
examine the resource annotations in session types and session skeletons.

Addressing this issue is beyond the scope of this article.
For simplicity, when facing a choice between branches, we sum resource
annotations in session skeletons to obtain the branches' respective cost bounds
and pick the branch with a higher cost bound.
The example \labelcref{eq:typing judgment of P in the example illustrating the
	difficulty of extracting a cost bound} still respects relative completeness of
worst-case input generation (\cref{theorem:relative completeness of worst-case
	input generation}) because the relative-completeness theorem require tight
cost bounds.


\newcommand{\costboundconfig}[1]{\lceil \bowtie \rceil  (#1)}
\newcommand{\colors}[0]{\calC}
\newcommand{\ids}[1]{\text{IDs}_{#1}} \newcommand{\cternary}[3]{(#1) \mathrel{?}
	#2 : #3}

\section{Worst-Case Input Generation Algorithm}
\label{sec:worst-case input generation}

Suppose we are given a SILL program $(P, \Sigma)$ and a collection $K$ of
skeletons for external channels.
The worst-case input generation algorithm is displayed in
\cref{algorithm:worst-case input generation for SILL}.
\begin{algorithm}
	\caption{Worst-case input generation algorithm for a SILL program}
	\label{algorithm:worst-case input generation for SILL}
	\begin{algorithmic}[1]
		\Procedure{WC Input Generation}{$(P, \Sigma), K$}
		\State Run AARA to infer resource-annotated session types of internal and
		external channels\label{aloline:AARA to infer resource annotations}
		\State Check that the skeletons $K$ satisfy all requirements
		\label{aloline:checking the suitability of skeletons}
		\State Run symbolic execution while tracking potential to identify an
		execution path where the cost bound is tight.
		Also, construct a path constraint $\phi$ of this execution path
		\label{aloline:symbolic exeuction}
		\State Solve the path constraint $\phi$ from the previous step
		\label{aloline:solving a path constraint}.
		If the path constraint is infeasible, go back to \cref{aloline:symbolic
			exeuction} to search for another execution path where the cost bound is
		tight.
		\EndProcedure
	\end{algorithmic}
\end{algorithm}

In \cref{aloline:AARA to infer resource annotations}, we run AARA to derive
resource annotations of both internal and external channels.
The resource annotations are used to keep track of potential during symbolic
execution (\cref{aloline:symbolic exeuction}).
In \cref{aloline:checking the suitability of skeletons}, we check the following
requirements for session skeletons of external channels:
\begin{itemize}
	\item The session skeletons $K$ are compatible with their original session
	      types (\cref{sec:compatibility of skeletons with session types}).
	\item The input portions of the session skeletons $K$ are finite
	      (\cref{sec:finite input portion of skeletons}).
\end{itemize}

\Cref{sec:checking tightness of cost bounds} describes how we track potential
during the symbolic execution (\cref{aloline:symbolic exeuction}).
\Cref{sec:symbolic exeuction algorithm formal presentation} formalizes the
symbolic execution.
If the symbolic execution finds an execution path where the cost bound is tight,
the corresponding path constraint $\phi$ is fed to an SMT solver in
\cref{aloline:solving a path constraint}.
If the path constraint $\phi$ is solvable, we obtain a concrete worst-case
input.
Otherwise, if the path constraint $\phi$ is infeasible (i.e., it has no
solutions), we go back to \cref{aloline:symbolic exeuction}, searching for
another execution path where the cost bound is tight.

\subsection{Checking the Tightness of Cost Bounds}
\label{sec:checking tightness of cost bounds}

The symbolic execution searches for an execution path where the cost bound is
tight.
Potential in SILL has a local nature: the potential is distributed across
processes and flows between them.
Hence, instead of tracking the total potential, we locally track individual
units of red potential (i.e., potential supplied by the external world) and blue
potential (i.e., potential freed up by $\m{tick} \; q$ for $q < 0$).

\paragraph{Red potential}

Red potential must eventually be consumed completely.
To check it, we equip each process with a Boolean flag $r \in \{\true,
	\false\}$.
The flag $r = \true$ means the process contains no red potential.
The flag is updated according to the following rules:
\begin{itemize}
	\item When (red) potential is supplied to the process by the external world,
	      we set $r = \false$.
	\item Suppose $q > 0$ units of potential are transferred from one process
	      (which initially has potential $p + q$ and a Boolean flag $r_1$) to
	      another process (which initially has a flag $r_2$).
	      The flag of the sender becomes $r_1 \lor (p = 0)$, and the flag of the
	      recipient becomes $r_1 \land r_2$.
\end{itemize}

%
We forbid processes from throwing away potential when $r = \false$.
Potential is thrown away by the rule $\m{relax}$ in the type system
(\cref{fig:some rules of the type system of resource-aware SILL}).
If it happens, it means the cost bound is not tight.
In such an event, the symbolic execution backtracks and explores another
execution path.
Likewise, red potential is not allowed to flow back to the external world, since
cost bounds only factor in incoming potential from the external world.

\paragraph{Blue potential}

Blue potential must be consumed if its generation precedes the consumption of
red potential.
To formally define what it means for an event (e.g., sending and receiving of
messages, and generation and consumption of potential) to precede another, we
introduce the happened-before relation $\to$ between events~\cite{Lamport1978},
which is a well-established notion in distributed and concurrent computing.

Pictorially, the happened-before relation is illustrated in
\cref{fig:high-water-mark cost in the presence of concurrency}.
The figure contains arrows within each of the processes $P_1$ and $P_2$ and
another arrow between them (for sending and receiving a message).
These arrows indicate the chronological ordering of events (regardless of how
concurrent processes are scheduled).
Taking the transitive closure of these arrows yields the happened-before
relation $\to$.

\begin{defi}[Happened-before relation~\cite{Lamport1978}]
	\label{def:happened-before relation}
	The happened-before relation $\to$ is the smallest binary relation between
	events (e.g., sending and receiving messages) that is closed under the
	following three conditions.
	First, if an event $A$ happens before an event $B$ on the same process, $A \to
		B$ holds, i.e., the event $A$ precedes the event $B$.
	Second, if an event $A$ sends a message and an event $B$ receives the message,
	$A \to B$ holds.
	Third, if $A \to B$ and $B \to C$ are true, so is $A \to C$.
\end{defi}

We now explain how to check whether blue potential, if generated before red
potential is consumed, is also consumed entirely.
Suppose that, during symbolic generation, blue potential is generated by an
expression $\m{tick} \; q$, where $q < 0$, in a process $P$.
We assign a fresh ID, say $ \textcolor{blue}{\text{blue}_{i}} \in \colors$, to
this newly created blue potential.
Here, $\colors = \{\textcolor{red}{\text{red}}\} \cup \{
	\textcolor{blue}{\text{blue}_{i}} \mid i \in \bbN\}$ is the set of all IDs for
red and blue potential.

\Cref{def:synchronization of two processes} defines what it means for one
process to be synchronized with another process since some event.

\begin{defi}[Synchronization of two processes]
	\label{def:synchronization of two processes}
	A process $Q$ is synchronized with another process $P$ since some event of the
	process $P$ if and only if the process $Q$'s current state happens after the
	event.
\end{defi}

For concreteness, suppose that a process $P$ has generated blue potential of the
ID $\textcolor{blue}{\text{blue}_1}$.
During the symbolic execution, we track the following items:
\begin{itemize}
	\item The set of processes that have been synchronized
	      (\cref{def:synchronization of two processes}) with the process $P$ since
	      the generation of the ID $\textcolor{blue}{\text{blue}_1}$.
	      We track such processes by passing around the ID
	      $\textcolor{blue}{\text{blue}_{1}}$ whenever a message is sent by the
	      process $P$ or other processes that carry the ID
	      $\textcolor{blue}{\text{blue}_{1}}$.
	      At any moment, the set of synchronized processes is given by the set of
	      processes carrying the ID $\textcolor{blue}{\text{blue}_{1}}$.
	\item What other potential is consumed by a process synchronized with the
	      process $P$.
	\item Whether the potential $\textcolor{blue}{\text{blue}_1}$ is completely
	      consumed.
\end{itemize}

To understand how tracking these items helps us detect loose cost bounds,
consider a scenario where
\begin{enumerate*}[label=(\roman*)]
	\item a process $Q$ has been synchronized with the process $P$;
	\item the process $Q$ consumes \textcolor{red}{red} potential; and
	\item the potential $\textcolor{blue}{\text{blue}_1}$ is not entirely consumed
	      by the end of program execution.
\end{enumerate*}
Because the potential $\textcolor{blue}{\text{blue}_1}$ is generated before the
potential \textcolor{red}{red} is consumed, we can substitute blue potential for
red potential, thereby lowering the amount of necessary red potential without
plunging into negative potential.
Hence, the cost bound is not tight in this case.

More generally, we can substitute blue potential for red potential along a chain
of distinct blue potential, each with a different ID.
For instance, we want to substitute certain blue potential (with ID, say,
$\textcolor{blue}{\text{blue}_1}$) for \textcolor{red}{red} potential, where the
potential $\textcolor{blue}{\text{blue}_1}$ was generated before the consumption
of \textcolor{red}{red} potential.
However, the potential $\textcolor{blue}{\text{blue}_1}$ is entirely consumed.
So we must find another blue potential (say $\textcolor{blue}{\text{blue}_2}$)
that can substitute for the potential $\textcolor{blue}{\text{blue}_1}$.
This is possible when
\begin{enumerate*}[label=(\roman*)]
	\item the potential $\textcolor{blue}{\text{blue}_2}$ is generated before the
	      potential $\textcolor{blue}{\text{blue}_1}$ is consumed and
	\item the potential $\textcolor{blue}{\text{blue}_2}$ is not entirely
	      consumed.
\end{enumerate*}
\Cref{fig:chain of the happened-before relation} depicts this situation where we
have two distinct potential, namely $\textcolor{blue}{\text{blue}_1}$ and
$\textcolor{blue}{\text{blue}_2}$, whose generation and consumption are related
by the happened-before relation.

\begin{figure}[t]
	\centering
	\includegraphics[scale=1]{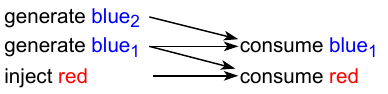}
	\caption{Alternating chain of the happened-before relation between events.
		The label ``inject \textcolor{red}{red}'' refers to an event of injecting
		red potential to some process from the external world.
		The label ``generate \textcolor{blue}{blue}'' refers to the generation of
		blue potential.}
	\label{fig:chain of the happened-before relation}
\end{figure}

If the potential $\textcolor{blue}{\text{blue}_2}$ is not completely consumed,
we can push the leftover potential of $\textcolor{blue}{\text{blue}_2}$, from
``generate $\textcolor{blue}{\text{blue}_2}$'' to ``inject
\textcolor{red}{red},'' along the alternating path of $\to$ in \cref{fig:chain
	of the happened-before relation}.
As a result, we can reduce the total \textcolor{red}{red} potential injected to
the program without plunging into negative potential, thereby lowering the cost
bound.

To correctly detect such alternating paths of the happened-before relation
$\to$, symbolic execution maintains a graph whose set of nodes is $\colors$
(i.e., set of IDs for red and blue potential).
The graph has an edge $(v_1, v_2)$ if and only if potential $v_2$ is consumed
after potential $v_1$ is generated.
In this graph, if there is a path from a node $\text{blue}_{i}$ to a node
$\text{red}$ such that the node $\text{blue}_{i}$ is not entirely consumed, then
the cost bound is not tight.
If the cost bound is detected to be loose on the current execution path, the
symbolic execution backtracks and explores another execution path.

\Cref{theorem:checking tightness of cost bounds} states the correctness of
tracking potential.
A proof is given in \cref{appendix:checking tightness of cost bounds}.

\begin{thm}[Checking tightness of cost bounds]
	\label{theorem:checking tightness of cost bounds}
	If red and blue potential is tracked without encountering issues, then the
	cost bound is tight.
\end{thm}

\subsection{Symbolic Execution}
\label{sec:symbolic exeuction algorithm formal presentation}

The symbolic execution runs a SILL program on a skeleton while keeping track of
potential.
Because the operational semantics of SILL is given by a multiset rewriting
system, we also use it to define the symbolic execution.
The symbolic execution involves two types of predicates:
\begin{equation}
	\label{eq:two predicates used in symbolic execution}
	\proc{\Delta; q \vdash P :: (c: A), \phi, \ids{}} \qquad \msg{c, M, \phi, \ids{}}.
\end{equation}

The first predicate of \cref{eq:two predicates used in symbolic execution}
represents a well-typed process $\Delta; q \vdash P :: (c: A)$, where $A$ is
either a resource-annotated skeleton (if $c$ is an external channel) or a
resource-annotated session type (if $c$ is an internal channel).
A logical formula $\phi$ is a path constraint so far and will later be fed to an
SMT solver.
The component $\ids{} = (\ids{s}, \ids{p})$ is a pair of finite sets of
potential's IDs.
The first set $\ids{s} \subset \{\textcolor{blue}{\text{blue}_{i}} \mid i
	\in\bbN \}$ tracks synchronization: if the set $\ids{s}$ of a process $P$
contains an ID $\textcolor{blue}{\text{blue}_{i}}$, then the process $P$'s
current state happens after the potential $\textcolor{blue}{\text{blue}_{i}}$
was generated.
The second set $\ids{p} \subset \colors$ tracks potential transfer: if the set
$\ids{p}$ of a process $P$ contains an ID $i$, it means the process $P$ contains
the (red or blue) potential identified by the ID $i$.

The second predicate $\msg{c, M, \phi, \ids{}}$ in \cref{eq:two predicates used
	in symbolic execution} represents a message $M$ (encoded as as process) that
provides a channel $c$.
A logical formula $\phi$ is a path constraint carried by the message.

\Cref{fig:key rules in the symbolic execution for the process layer} displays
key rewriting rules.
The remaining rules are given in \cref{appendix:symbolic execution}.

\begin{figure}[t]
	\begin{small}
		\begin{mathpar}
			\inferrule
			{\proc{\Delta, c: b \supset A; p \vdash \m{send} \; c \; e; P :: (d: D),
					\phi_1, (\ids{s}, \ids{p})} \\ e \Downarrow \tuple{\phi_2, v} \\
				c' \text{ is fresh}} {\proc{\Delta, c': A; p \vdash P [c' / c] :: (d:
					D), \phi_1, (\ids{s}, \ids{p})} \\ \msg{c', \m{send} \; c \; v; c'
					\leftarrow c, \phi_2, (\ids{s}, \emptyset)}} {\supset} S
			\and
			\inferrule
			{\proc{\Delta; p \vdash \m{case} \; c \; \{\ell_{i} \hookrightarrow P_{i}
					\mid i \in N\} :: (c: \&_{x} \{\ell_i: A_i \mid i \in N \}), \phi,
					(\ids{s}, \ids{p})} \\ A_k \text{ has the highest cost bound}}
			{\proc{\Delta; p \vdash P_{k} [c' / c] :: (c': A_k), \phi \land (x =
					k), (\ids{s}, \ids{p})}} {\&} R_{\text{external}}
			\and
			\inferrule
			{\proc{\cdot; p \vdash \m{close} \; c :: (c: \mathbf{1}), \phi, (\ids{s},
					\ids{p})} \\ \textcolor{red}{\text{red}} \notin \ids{p}} {\msg{c,
					\m{close} \; c, \phi, (\ids{s}, \emptyset)}} \mathbf{1} S
			\and
			\inferrule
			{\proc{\Delta, c: \triangleleft^{q} A; p+q \vdash \m{pay} \; c \; \{q\}; P
					:: (d: B), \phi, (\ids{s}, \ids{p})} \\\\ c' \text{ is fresh} \\
				\textcolor{red}{\text{red}} \notin \ids{p}} {\proc{\Delta, c': A; p
					\vdash P [c' / c] :: (d: B), \phi, (\ids{s}, \cternary{p =
						0}{\emptyset}{\ids{p}} ) } } {\triangleleft} L_{\text{external}}
			\and
			\inferrule
			{\proc{\Delta; p \vdash \m{get} \; c \; \{q\}; P :: (c: \triangleleft^{q}
					A), \phi, (\ids{s}, \ids{p})}} {\proc{\Delta; p+q \vdash P :: (c: A),
					\phi, (\ids{s}, \ids{p} \cup \{ \textcolor{red}{\text{red}} \}) }}
			{\triangleleft} R_{\text{external}}
			\and
			\inferrule
			{\proc{\Delta; p + q \vdash \m{tick} \; q; P :: (c: A), \phi, (\ids{s},
					\ids{p})}} {\proc{\Delta; p \vdash P :: (c: A), \phi, (\ids{s},
					\cternary{p = 0}{\emptyset}{\ids{p}} ) }} \m{tick}_{> 0}
			\and
			\inferrule
			{\proc{\Delta; p \vdash \m{tick} \; (- q); P :: (c: A), \phi, (\ids{s},
					\ids{p})} \\ \textcolor{blue}{\text{blue}_{i}} \in \colors \text{ is
					fresh}} {\proc{\Delta; p + q \vdash P :: (c: A), \phi, (\ids{s} \cup
					\{\textcolor{blue}{\text{blue}_{i}} \}, \ids{p} \cup \{
					\textcolor{blue}{\text{blue}_{i}} \} ) }} \m{tick}_{< 0}
		\end{mathpar}
	\end{small}
	\caption{Key rules in the process-layer symbolic execution.
		Throughout the rules, we have $q > 0$.
		A judgment $e \Downarrow \tuple{\phi, v}$ means a functional term $e$
		evaluates to a (symbolic) value $v$ with a path constraint $\phi$.
		A ternary operator $\cternary{b}{e_1}{e_2}$ returns $e_1$ if the Boolean
		value $b$ evaluates to true and $e_2$ otherwise.}
	\label{fig:key rules in the symbolic execution for the process layer}
\end{figure}

In the rule ${\supset} S$, when a message is sent, it also carries the set
$\ids{s}$ of the sender.
It is then added to the set $\ids{s}$ of the recipient.

The rule ${\&} R_{\text{external}}$ resolves an external choice on an external
channel.
A label $k \in N$ is chosen such that the skeleton $A_k$ has the highest cost
bound among $\{A_i \mid i \in N\}$.
The path constraint $\phi$ is then augmented with a constraint $x = k$,
indicating that the external world should choose $k \in N$ to trigger the
worst-case behavior.
If we find out later that the current execution path's cost bound is not tight,
we backtrack and try a different $k' \neq k$ such that $A_{k'}$ has the highest
cost bound.
If the algorithm fails to find any $A_k$ ($k \in N$) with the highest cost bound
that is tight, then algorithm returns no worst-case inputs.

The rule $\mathbf{1}_{S}$ forbids red potential from being wasted: red potential
must not remain when a process terminates.
Likewise, the rule $\triangleleft L_{\text{external}}$, which is for an external
channel, forbids red potential from flowing back to the external world.
Furthermore, if we waste blue potential (i.e., $\ids{p} \setminus \{
	\textcolor{red}{\text{red}} \} \neq \emptyset$) in the rules $\mathbf{1}_{S}$
and $\triangleleft L_{\text{external}}$, we must record its IDs because we do
not want to waste blue potential that could have been substituted for red
potential.

In the rule $\triangleleft R_{\text{external}}$, the process receives red
potential from the external world.
So the ID $\textcolor{red}{\text{red}}$ is added to the set $\ids{p}$ of the
recipient.

Finally, we have two rules for the construct $\m{tick}$.
In the rule $\m{tick}_{> 0}$, the potential stored in the process is consumed.
Whenever this rule is applied, we must record the pair $(\ids{s}, \ids{p})$.
This pair indicates which blue-potential ID was generated before red potential
is consumed.
In the rule $\m{tick}_{< 0}$, blue potential is generated.
Hence, we generate a fresh ID and add it to the sets $\ids{s}$ and $\ids{p}$.

\subsection{Soundness and Relative Completeness}

\Cref{theorem:checking tightness of cost bounds} assures us that symbolic
execution's high-level idea is sound.
However, it assumes that the symbolic execution algorithm correctly tracks costs
and potential.
To prove this assumption, we show a simulation between the symbolic execution
and cost semantics.
The simulation then leads to the soundness and relative completeness of our
worst-case input generation algorithm.

\paragraph{Cost bounds encoded by resource annotations}

\Cref{def:cost bounds of skeletons} defines the cost bound of a
resource-annotated skeleton.

\begin{defi}[Cost bounds of skeletons]
	\label{def:cost bounds of skeletons}
	Given a skeleton $K$ for an external channel provided by a process, its cost
	bound, denoted by $\costboundleft{K}$, is defined as follows.
	\begin{align*}
		\costboundleft{H \supset K}       & \coloneq \costboundleft{K}                           & \costboundleft{\&_{x} \{\ell_{i}: K_{i} \mid i \in N \}} & \coloneq \max_{i \in N} \costboundleft{K_{i}} \\
		\costboundleft{b \land K}         & \coloneq \costboundleft{K}                           & \costboundleft{\oplus \{\ell_{i}: K_{i} \mid i \in N \}} & \coloneq \max_{i \in N} \costboundleft{K_{i}} \\
		\costboundleft{K_1 \multimap K_2} & \coloneq  \costboundright{K_1} + \costboundleft{K_2} & \costboundleft{\triangleleft^{q} K}                      & \coloneq q + \costboundleft{K}                \\
		\costboundleft{K_1 \otimes K_2}   & \coloneq  \costboundleft{K_1} + \costboundleft{K_2}  & \costboundleft{\triangleright^{q} K}                     & \coloneq \costboundleft{K}                    \\
		\costboundleft{\mathbf{1}}        & \coloneq 0.
	\end{align*}
	The dual cost bound $\costboundright{\cdot}$, which is used in the definition
	of $\costboundleft{K_1 \multimap K_2}$, is defined analogously.
\end{defi}


In order for \cref{def:cost bounds of skeletons} to make sense, the input
portion of the skeleton $K$ must be finite.
$\triangleleft^{q}$ contributes to cost bounds in \cref{def:cost bounds of
	skeletons}, but $\triangleright^{q}$ does not.
If the cost bounds factored in outgoing potential (i.e., $\triangleright^{q}$)
as well as incoming potential (i.e., $\triangleleft^{q}$), the cost bounds might
depend on the output size.
In skeletons, while the input size is fixed, the output size is not known
statically.
Consequently, before looking for an execution path with a tight cost bound, the
worst-case input generation algorithm would first need to maximize the cost
bound (by minimizing the output size).
As this will complicate the algorithm, we do not factor outgoing potential into
cost bounds.

\begin{defi}[Cost bounds of configurations]
	\label{def:cost bounds of configurations}
	Suppose $C$ is a configuration with external channels $c_1, \ldots, c_m,
		c_{m+1}, \ldots, c_{n}$.
	Here, the channels $c_1, \ldots, c_{m}$ are provided by processes in the
	configuration $C$, whereas the channels $c_{m+1}, \ldots, c_{n}$ are provided
	by the external world.
	Let $K_1, \ldots, K_{n}$ be the skeletons of external channels.
	Also, let $p$ be the total potential locally stored in processes in the
	configuration $C$.
	The cost bound of the configuration $C$ is defined as
	\begin{equation}
		\costboundconfig{C} \coloneq p + \sum_{i = 1}^{m} \costboundright{K_{i}} + \sum_{i = m+1}^{n} \costboundleft{K_{i}}.
	\end{equation}
	Here, $\costboundright{K_{i}}$ denotes the cost bound of a skeleton $K_{i}$
	when the external channel is provided by some process in the network
	(\cref{def:cost bounds of skeletons}).
	The dual cost bound is $\costboundleft{K_{i}}$.
\end{defi}

\paragraph{Similarity relation and simulation}

\Cref{def:similarity between predicates} defines a similarity relation between
predicates.
This relation can be lifted from predicates to configurations
(\cref{def:similarity between configurations}).

\begin{defi}[Similarity between predicates]
	\label{def:similarity between predicates}
	Fix $S$ to be a solution (i.e., a mapping from skeleton variables to concrete
	values) to a path constraint generated by the symbolic execution.
	The similarity relation $\sim$ between a predicate in symbolic execution and a
	predicate in the cost semantics is defined by
	\begin{small}
		\begin{mathpar}
			\inferrule
			{S \vdash P_{\text{sym}} = P_{\text{cost}}} {\proc{\uscore; \uscore \vdash
					P_{\text{sym}} :: (c: \uscore), \uscore, \uscore} \sim \proc{c,
					P_{\text{cost}}}}
			\and
			\inferrule
			{S \vdash M_{\text{sym}} = M_{\text{cost}}} {\msg{c, M_{\text{sym}},
					\uscore} \sim \msg{c, M_{\text{cost}}}}.
		\end{mathpar}
	\end{small}
	Here, the premise $S \vdash P_{\text{sym}} = P_{\text{cost}}$ means a
	predicate $P_{\text{sym}}$ in the symbolic execution and a predicate
	$P_{\text{cost}}$ in the cost semantics are identical under the mapping $S$.
\end{defi}

\Cref{prop:simulation for soundness} establishes a simulation between the
symbolic execution and cost semantics.
A proof is given in \cref{appendix:soundness and relative completeness}.

\begin{prop}[Simulation for soundness]
	\label{prop:simulation for soundness}
	Suppose we are given three configurations: $C_{1, \text{sym}}$, $C_{2,
				\text{sym}}$, and $C_{1, \text{cost}}$.
	The first two configurations are used in the symbolic execution, and the last
	one is used in the cost semantics.
	These configurations satisfy two conditions:
	\begin{enumerate*}[label=(\roman*)]
		\item $C_{1, \text{sym}}$ transitions to $C_{2, \text{sym}}$ in one step of
		      the symbolic execution and
		\item $C_{1, \text{sym}} \sim C_{1, \text{cost}}$ holds.
	\end{enumerate*}
	Then there exists a configuration $C_{2, \text{cost}}$ of the cost semantics
	such that the following diagram commutes:
	\begin{equation}
		\label{eq:commutative diagram for the simulation for soundness}
		\xymatrix{
		C_{1, \text{sym}} \ar[r]_{w} & C_{2, \text{sym}} \\
		C_{1, \text{cost}} \ar@{~}[u] \ar[r]_{w}^>{\leq 1} & C_{2, \text{cost}} \ar@{~}[u]
		}
	\end{equation}
	In the commutative diagram \labelcref{eq:commutative diagram for the
		simulation for soundness}, $C_{1, \text{sym}} \xrightarrow[w]{} C_{2,
				\text{sym}}$ means $\costboundconfig{C_{1, \text{sym}}} -
		\costboundconfig{C_{2, \text{sym}}} = w$, where $\costboundconfig{\cdot}$
	denotes a cost bound of a configuration (\cref{def:cost bounds of
		configurations}).
	Likewise, $C_{1, \text{cost}} \xrightarrow[w]{} C_{2, \text{cost}}$ means the
	configuration $C_{1, \text{cost}}$ transitions to $C_{2, \text{cost}}$ such
	that the net cost increases by $w$.
	The arrow $\rightarrow^{\leq 1}$ means the number of steps is either zero or
	one.
\end{prop}

In \cref{prop:simulation for soundness}, one transition step in symbolic
execution may correspond to zero steps in the cost semantics.
It happens when we transfer potential by rewriting rules such as
${\triangleleft} L_{\text{external}}$ in the symbolic execution, which do not
have corresponding rules in the cost semantics.

\Cref{theorem:soundness of worst-case input generation} states the soundness of
the worst-case input generation algorithm.
In the statement of the theorem, $\interp{K_1, \ldots, K_n}$ denotes the set of
all possible inputs conforming to session skeletons $K_1, \ldots, K_n$ on
channels $c_1, \ldots, c_n$.
Each input is a multiset of predicates $\msg{c, M}$ for channels $c$ and
messages $M$.
A formal definition of the set of possible inputs is given in \cref{def:set of
	possible inputs}.

\begin{thm}[Soundness of worst-case input generation]
	\label{theorem:soundness of worst-case input generation}
	Given a collection $K_1, \ldots, K_n$ of skeletons, suppose the symbolic
	execution algorithm successfully terminates.
	Let $\phi$ be a path constraint generated by the symbolic execution and $t \in
		\interp{K_1, \ldots, K_n}$ be an input satisfying $\phi$.
	Then $t$ has the highest high-water-mark cost of all inputs from $\interp{K_1,
			\ldots, K_n}$.
\end{thm}

\begin{proof}
	\Cref{prop:simulation for soundness} shows that the symbolic execution
	correctly tracks of potential and net cost: both of them change by the same
	amount (but in the opposite direction).
	Furthermore, by \cref{theorem:checking tightness of cost bounds}, the
	high-water-mark cost bound is tight for the solution $t$ to the path
	constraint $\phi$.
	Therefore, the high-water-mark cost of the input $t$ to the SILL program is
	indeed the highest among all possible inputs $\interp{K_1, \ldots, K_n}$.
\end{proof}

Finally, \cref{theorem:relative completeness of worst-case input generation}
states the relative completeness of the algorithm.

\begin{thm}[Relative completeness of worst-case input generation]
	\label{theorem:relative completeness of worst-case input generation}
	Given a SILL program and a collection $K$ of resource-annotated session
	skeletons, assume the following:
	\begin{enumerate}[label=(A\arabic*),ref=A\arabic*]
		\item The processes are typable in resource-aware SILL.
		      \label{assumption:the proceses are typable}
		\item The cost bound of $K$ is tight on some finitely long execution path of
		      the SILL program. \label{assumption:the cost bound is tight on a finite executon path}
		\item The background theory for path constraints is decidable.
	\end{enumerate}
	Then the worst-case input generation algorithm \cref{algorithm:worst-case
		input generation for SILL} returns a valid worst-case input.
\end{thm}

\begin{proof}
	The assumption \labelcref{assumption:the proceses are typable} is necessary
	because session skeletons, which define the shape of a worst-case input to be
	synthesized and guide the symbolic execution, are based on resource-annotated
	session types.
	Hence, all channels in the SILL program must be typable with
	resource-annotated session types.

	Under the assumption \labelcref{assumption:the cost bound is tight on a finite
		executon path}, there exists some finitely long execution path $\pi$ in the
	SILL program that has the same high-water mark as the cost bound encoded in
	the resource-annotated skeletons $K$.
	The symbolic execution exhaustively searches for an execution path where the
	cost bound is tight, until a suitable execution path is found or all execution
	paths fail (\cref{algorithm:worst-case input generation for SILL}).
	Hence, we will eventually find the target execution path $\pi$, thanks to
	\cref{prop:simulation for soundness} stating the correctness of tracking the
	cost during the symbolic execution.
	The symbolic execution has a risk of non-termination on some execution paths.
	To avoid being stuck in the exploration of one execution path, the symbolic
	execution can explore multiple execution paths in parallel such that any
	finite execution path will eventually be explored.
\end{proof}


\section{Case Study: Web Server and Browsers}
\label{sec:case study}

We model the interaction between a web server and multiple web browsers.
For each browser, a new channel is spawned, and the browser first engages in a
three-way handshake protocol with the server (as in TCP).
Once the handshake protocol is successfully completed, the browser and server
proceed to the main communication phase for data transfer.

This case study considers the non-monotone resource metric of memory.
Suppose the server requires
\begin{enumerate*}[label=(\roman*)]
	\item one memory cell for the handshake protocol and
	\item another memory cell for the subsequent communication after the
	      handshake.
\end{enumerate*}
The first memory cell stores so-called sequence numbers that are established
during the handshake.
The second memory cell for the main communication phase stores data about a
browser.

Let $c$ be a channel provided by the server and used by browsers.
Without loss of generality, suppose we have two browsers that want to
communicate with the server.
We examine two implementations of the server.
In the first implementation (\cref{sec:independent sessions in the case study}),
the two browsers' sessions run independently of each other.
So the server cannot control how the sessions are scheduled.
By contrast, in the second implementation (\cref{sec:coordinating sessions with
	a scheduler in the case study}), the server coordinates sessions with the help
of a scheduler.

\subsection{Independent Sessions}
\label{sec:independent sessions in the case study}

\begin{figure}[t]
	\centering
	\includegraphics[width=0.65\columnwidth]{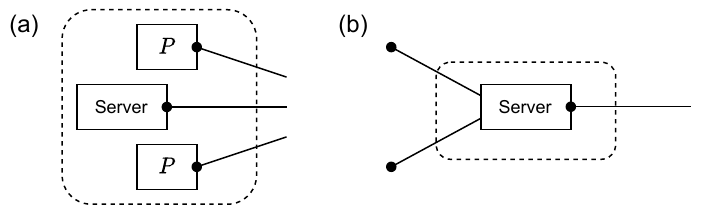}
	\caption{(a) First model of a web server.
		Here, an independent process $P$ is assigned to each browser's session that
		wants to communicate with the server.
		(b) Second model of a web server where a scheduler is modeled.}
	\label{fig:case study}
\end{figure}

In the first implementation, the session type of the provided channel is
\begin{equation}
	c: A \otimes A \otimes \mathbf{1}.
\end{equation}
That is, the server sequentially spawns a new channel of session type $A$ for
each of the two browsers.
The server after spawning two channels is depicted in \cref{fig:case study} (a).
The channel providers $P$ of these channels run independently of each
other---the server cannot control the order of events on these channels.

Resource-annotated session type $A$ is defined as
\begin{equation}
	A \coloneq \integer \supset \textcolor{red}{\triangleleft^{1}} (\integer \times \integer) \land  \& \{ \m{ack}: \integer \supset \oplus \{\m{success}: \textcolor{red}{\triangleleft^{1}} \mathbf{1}, \m{failure}: \mathbf{1} \}, \m{timeout}: \mathbf{1} \}.
\end{equation}
In the session type $A$, the browser first sends an integer $x$, and the server
sends back two integers: $1 + x$ and $y$.
In response, the browser sends either a label $\m{ack}$ followed by an integer
(ideally $1 + y$) or a label $\m{timeout}$ to indicate that the browser is
inactive.

After successful completion of the handshake, the server sends a label
$\m{success}$ and terminates.
If the browser sends back a wrong integer (i.e., an integer different from $1 +
	y$) to the server, the server sends a label $\m{failure}$.
After $\m{success}$ and $\m{failure}$, there is no further communication between
the server and browser.
This is for simplicity in our modeling.
In practice, the $\m{success}$ branch has further communication.
The detailed implementation of such a server is given in \cref{eq:implementation
	of the server where channels are created by the server} (\cref{appendix:case
	study}).

Let a session skeleton for the external channel $c$ be $K_1 \otimes K_2 \otimes
	\mathbf{1}$ , where the session skeleton $K_{i}$ ($i = 1, 2$) is
\begin{equation}
	K_{i} \coloneq z_{i,1} \supset \textcolor{red}{\triangleleft^{1}} (\integer \times \integer) \land \& \{ \m{ack}: z_{i,2} \supset \oplus \{\m{success}: \textcolor{red}{\triangleleft^{1}} \mathbf{1}, \m{failure}: \mathbf{1} \}, \m{timeout}: \mathbf{1} \}.
\end{equation}
Here, $z_{i,1}$ and $z_{i,2}$ are skeleton variables of the integer type.
Thus, the skeleton $K_1 \otimes K_2 \otimes \mathbf{1}$ is simply obtained from
the original session type $A \otimes A \otimes \mathbf{1}$ by substituting fresh
skeleton variables for all occurrences of $\integer$ in the session type.

The worst-case cost bound of the skeleton $K_1 \otimes K_2 \otimes \mathbf{1}$
is 4.
This is because each $K_i$ ($i = 1, 2$) has a worst-case cost bound of 2, which
happens if the browser sends the label $\m{ack}$ to the server.
The worst-case input generation algorithm generates a worst-case
input
\begin{equation}
	\label{eq:worst-case input for the first implementation in the case study}
	z_{1,1} = x_1 \quad z_{1,2} = y_1 + 1 \quad z_{2,1} = x_2 \quad z_{2,2} = y_2 + 1,
\end{equation}
where $x_1, x_2 \in \bbZ$ are unconstrained, and $y_1, y_2 \in \bbZ$ are the
integers sent back from the server to a browser.
The constants $y_1, y_2$ are assumed to be hard-coded in the server's code.

More generally, if we have $n$ many browsers, the cost bound becomes $2 n$.
It is tight: because different sessions run independently, in the worst case, we
will need $2 n$ memory cells at the peak memory usage.
As before, the worst-case input generation algorithm can compute a correct
worst-case input of high-water-mark cost $2 n$.
It suggests that adversaries can overwhelm the server's memory by deploying a
large number of browsers.

This vulnerability is reminiscent of a denial-of-service (DoS) attack.
However, our SILL implementation does not faithfully model the true DoS attack.
A DoS attack creates a large number of so-called half-open sessions, where the
completion of the handshake is delayed by withholding the label $\m{ack}$.
Meanwhile, our SILL implementation does not model the withholding of the label
$\m{ack}$.
Instead, all browsers' sessions will eventually terminate in our SILl
implementation, and the worst-case high-water-mark cost of the entire concurrent
program happens when the high-water marks of all individual processes coincide.
To faithfully model a DoS attack, it would be necessary to model the withholding
of sending the label $\m{ack}$ in a large number browsers' sessions at the same.
However, it cannot be faithfully modeled in SILL due to the following
limitations in the expressivity:
\begin{enumerate*}[label=(\roman*)]
	\item SILL cannot withhold events (e.g., sending and receiving
	      messages)---they must proceed; and
	\item SILL cannot specify the exact timing of events, such as that certain
	      events on different channels and processes all happen at the same time.
\end{enumerate*}

\subsection{Coordinating Sessions with Schedulers}
\label{sec:coordinating sessions with a scheduler in the case study}

Now consider an alternative implementation where it is the external world that
spawns channels.
The typing judgment of the channel $c$ is
\begin{equation}
	c: A \multimap A \multimap \mathbf{1},
\end{equation}
where the session type $A$ (without resource annotations) is
\begin{equation}
	\label{eq:annotation-free A in the second implementation of the case study}
	A \coloneq \integer \land (\integer \times \integer) \supset \oplus \{ \m{ack}: \integer \land \& \{\m{success}: \mathbf{1}, \m{failure}: \mathbf{1} \}, \m{timeout}: \mathbf{1} \}.
\end{equation}
The resource annotation in the session type $A$ depends on how the browsers'
sessions are scheduled.

Unlike in \cref{sec:independent sessions in the case study}, in this section,
the two channels are directly connected to the server (\cref{fig:case study}
(b)).
Hence, the server can/must coordinate the communication sessions on the
channels.
For example, the server can use a round-robin scheduler that alternates between
the two browsers.
Another possibility is a sequential scheduler: the server serves the first
browser and then moves on to the second one after the first browser is finished.

With a round-robin scheduler, we obtain the same high-water mark as
\cref{sec:independent sessions in the case study}.
With a sequential scheduler, the typing judgment of channel $c$ is
\begin{equation}
	c: A_\text{anno} \multimap A \multimap \mathbf{1},
\end{equation}
where the resource-annotated session type for the first browser is
\begin{equation}
	A_{\text{anno}} \coloneq \integer \land \textcolor{red}{\triangleright^{1}} (\integer \times \integer) \supset  \oplus \{ \m{ack}: \integer \land \& \{\m{success}: \textcolor{red}{\triangleright^{1}} \mathbf{1}, \m{failure}: \mathbf{1} \}, \m{timeout}: \mathbf{1} \}.
\end{equation}
and the session type $A$ for the second browser, which happens to require zero
potential, is given in \cref{eq:annotation-free A in the second implementation
	of the case study}.
It is a valid typing judgment because once the server finishes talking with the
first browser, two memory cells are freed up and are reused for the second
browser.
Therefore, the sequential scheduler's cost bound is lower than the round-robin
scheduler's.

More generally, if we have $n$ browsers, the sequential scheduler's cost bound
remains 2.
Further, thanks to the soundness of resource-annotated session types, the bound
2 is a valid cost bound.
Therefore, adversaries cannot overwhelm the server by sending a large number of
communication requests.


\section{Related Work}
\label{sec:related work}

\paragraph{Resource analysis}

Resource analysis of programs aims to derive symbolic cost bounds.
Numerous approaches exist: type
systems~\cite{Crary2000,Vasconcelos2008,Danielsson2008,Lago2011,Avanzini2017,Cicek2017,Handley2019},
recurrence
relations~\cite{Wegbreit1975,Grobauer2001,Albert2007,Kincaid2017popl,Kavvos2019,Cutler2020},
term rewriting~\cite{Avanzini2013,Brockschmidt2014,Hofmann2014,Moser2020}, and
static analysis~\cite{Gulwani2009speedpopl,Avanzini2015,Chatterjee2019}.
Among type-based approaches is AARA.
Linear AARA was first developed by Hofmann and Jost~\cite{Hofmann2003} and later
extended to univariate polynomial bounds~\cite{Hoffmann2010}, multivariate
polynomial bounds~\cite{Hoffmann2012}, and exponential bounds~\cite{Kahn2020}.
AARA has also been incorporated into imperative
programming~\cite{Carbonneaux2015}, parallel programming~\cite{Hoffmann2015},
and probabilistic programming~\cite{Ngo2018,Wang2020,Avanzini2020}.

\paragraph{Session types}

Session types, which describe communication protocols on channels, were
originally proposed by Honda~\cite{Honda1993}.
Caires and Pfenning~\cite{Caires2010} build a session type system whose logical
counterpart is intuitionistic linear logic.
Their session type system  has been integrated into a functional programming
language using contextual monads, resulting in the language
SILL~\cite{Toninho2013,Pfenning2015}.
Resource-aware SILL~\cite{Das2018} incorporates linear AARA into SILL (excluding
shared channels).
Nomos~\cite{Das2019} is a session-typed programming language that uses
resource-aware SILL to infer gas bounds of smart contracts.
Wadler has developed a session type system based on classical linear
logic~\cite{Wadler2012}.
Binary session types have also been extended to multiparty
ones~\cite{Honda2008}.

\paragraph{Worst-case input generation}

The present work was inspired by the type-guided worst-case input generation for
functional programming by Wang and Hoffmann~\cite{Wang2019}.
While \cite{Wang2019} focuses on monotone resource metrics in sequential
programming, the present work considers more general non-monotone resource
metrics in message-passing concurrent programming.
Non-monotone resource metrics, when combined with concurrency of processes, pose
challenges to worst-case input generation.

WISE~\cite{Burnim2009} is the first work to use symbolic execution for
worst-case input generation.
It first explores an entire search space for a small input to identify a
worst-case execution path.
This path is then generalized to a branch policy to handle larger inputs.
The use of branch policies reduces the search space of large inputs, thereby
making worst-case input generation more scalable.
SPF-WCA~\cite{Luckow2017} extends WISE with path policies that take into account
histories when we determine which branch to take during symbolic execution.

Instead of branch and path policies, Wang and Hoffmann~\cite{Wang2019} and we
use resource-annotated types to guide symbolic execution.
One advantage of type-guided symbolic execution is that worst-case input
generation becomes sound.
Another advantage is that we learn how many non-monotone resources can be
transferred between concurrent processes.

The fuzzing research community has investigated worst-case input generation.
SlowFuzz~\cite{Petsios2017} is the first fuzzer that automatically finds
worst-case inputs with gray-box access to programs.
PerfFuzz~\cite{Lemieux2018} extends SlowFuzz with multi-dimensional objectives.
MemLock~\cite{Wen2020} focuses on memory consumption bugs.
Although fuzzing generally offers neither soundness nor relative completeness,
it is more scalable than static-analysis-based worst-case input generation
because fuzzers do not analyze programs' internal workings.


\section{Conclusion}
\label{sec:conclusion}

It is non-trivial to generate worst-case inputs to concurrent programs under
non-monotone resource metrics.
The high-water-mark cost of a concurrent program depends on how processes are
scheduled at runtime.
As a result, haphazardly executing a concurrent program may not reveal its
correct high-water-mark cost.

In this work, we have developed the first sound worst-case input generation
algorithm for message-passing concurrent programming under non-monotone resource
metrics.
The key insight is to have resource-annotated session types guide symbolic
execution.
We have also identified several technical challenges posed by session types in
the design of skeletons.
We have proved the soundness and relative completeness of our algorithm.
Finally, we have presented a simple case study of a web server's memory usage,
illustrating the utility of the worst-case input generation algorithm.

\bibliographystyle{alphaurl}
\bibliography{
	./bib/aara,
	./bib/resource_analysis,
	./bib/relational_analysis,
	./bib/mechanism_design,
	./bib/miscellaneous,
	./bib/session_types,
	./bib/smart_contracts,
	./bib/input_generation}

\appendix


\section{Resource-Aware SILL}

\subsection{Cost Semantics and the Type System}
\label{appendix:cost semantics and the type system}

\Cref{fig:some rules of the cost semantics of SILL} already gives some rewriting
rules of the substructural cost semantics of SILL.
Remaining rules are displayed in \cref{fig:remaining rules of the cost semantics
of SILL}.

\begin{figure}[t]
	\begin{small}
		\begin{mathpar}
			\inferrule
			{\msg{c, w_1, \m{send} \; c \; v; c \leftarrow c'} \\ \proc{d, w_2, x
					\leftarrow \m{recv} \; c; P_{x}}} {\proc{d, w_1 + w_2, P_{v} [c' / c]}}
			{\land} R
			\and
			\inferrule
			{\proc{c, w, \m{send} \; c \; e; P} \\ e \Downarrow v \\ c' \text{ is
					fresh}} {\proc{c', w, P [c' / c]} \\ \msg{c, 0, \m{send} \; c \; v; c
					\leftarrow c'}} {\land} S
			\and
			\inferrule
			{\proc{d, w, \m{send} \; c_1 \; c_2; P} \\ c_1' \text{ is fresh}}
			{\proc{d, w, P [c_1' / c_1]} \\ \msg{c_1', 0, \m{send} \; c_1 \; c_2; c_1'
					\leftarrow c_1}} {\multimap} S
			\and
			\inferrule
			{\msg{c_1', w_1, \m{send} \; c_1 \; c_2; c_1' \leftarrow c_1} \\\\
				\proc{c_1, w_2, x \leftarrow \m{recv} \; c_1; P_{x}}} {\proc{c_1, w_1 +
					w_2, P_{c_2} [c_1' / c_1]}} {\multimap} R
			\and
			\inferrule
			{\msg{c_1, w_1, \m{send} \; c_1 \; c_2; c_1 \leftarrow c_1'} \\\\ \proc{d,
					w_2, x \leftarrow \m{recv} \; c_1; P_{x}}} {\proc{d, w_1 + w_2, P_{c_2}
						[c_1' / c_1]}} {\otimes} R
			\and
			\inferrule
			{\proc{c_1, w, \m{send} \; c_1 \; c_2; P} \\ c_1' \text{ is fresh}}
			{\proc{c_1', w, P [c_1' / c_1]} \\ \msg{c_1, 0, \m{send} \; c_1 \; c_2;
					c_1 \leftarrow c_1'}} {\otimes} S
			\and
			\inferrule
			{\proc{d, w, c.\ell_{k}; P} \\ c' \text{ is fresh}} {\proc{d, w, P [c' /
							c]} \\ \msg{c', 0, c.\ell_{k}; c' \leftarrow c}} {\&} S
			\and
			\inferrule
			{\msg{c', w_1, c. \ell_{k}; c' \leftarrow c} \\\\ \proc{c, w_2, \m{case}
					\; c \; \{ \overline{\ell_{i} \hookrightarrow P_{i}} \}}} {\proc{c, w_1 +
					w_2, P_{k} [c' / c]}} {\&} R
			\and
			\inferrule
			{\msg{c, w_1, c. \ell_{k}; c \leftarrow c'} \\\\ \proc{d, w_2, \m{case} \;
					c \; \{ \overline{\ell_{i} \hookrightarrow P_{i}} \}}} {\proc{d, w_1 +
					w_2, P_{k} [c' / c]}} {\oplus} R
			\and
			\inferrule
			{\proc{c, w, c.\ell_{k}; P} \\ c' \text{ is fresh}} {\proc{c', w, P [c' /
							c]} \\ \msg{c, 0, c.\ell_{k}; c \leftarrow c'}} {\oplus} S
			\and
			\inferrule
			{\proc{c_1, w, c_1 \leftarrow c_2}} {\msg{c_1, w, c_1 \leftarrow c_2}}
			\m{fwd}_{s}
			\and
			\inferrule
			{\proc{c_2, w_1, P} \\ \msg{c_1, w_2, c_1 \leftarrow c_2}} {\proc{c_1, w_1
					+ w_2, P [c_1 / c_2]}} \m{fwd}_{r}^{+}
			\and
			\inferrule
			{\msg{c_1, w_1, c_1 \leftarrow c_2} \\\\ \proc{d, w_2, P}} {\proc{d, w_1 +
					w_2, P [c_2 / c_1]}} \m{fwd}_{r}^{-}
		\end{mathpar}
	\end{small}
	\caption{Remaining rules in the cost semantics of SILL. The judgment $e
			\Downarrow v$ means functional term $e$ evaluates to value $v$.}
	\label{fig:remaining rules of the cost semantics of SILL}
\end{figure}

To treat $\proc{\cdot}$ and $\msg{\cdot}$ uniformly, message $M$ in predicate
$\msg{c, M}$ is encoded as a process~\cite{Das2018}, just like $P$ in predicate
$\proc{c, P}$.
The grammar of message $M$, which is subsumed by the grammar of processes, is
presented below.
\begin{align*}
	M \Coloneq {} & c \leftarrow c' \mid c.\ell_{i}; c \leftarrow c' \mid c.\ell_{i}; c' \leftarrow c                                        \\
	              & \mid \m{send} \; c \; v; c \leftarrow c' \mid \m{send} \; c \; v; c' \leftarrow c                                        \\
	              & \mid \m{send} \; c_1 \; c_2; c_1 \leftarrow c_1' \mid \m{send} \; c_1 \; c_2; c_1' \leftarrow c_1 \mid \m{close} \; c.
\end{align*}
Here, $c$ and $c'$ are channels, $v$ is a functional-layer value, and $q \in
	\bbQ_{> 0}$ is a quantity of potential.

Some rules of the type system of resource-aware SILL are already given in
\cref{fig:some rules of the cost semantics of SILL}.
The remaining rules are presented in \cref{fig:remaining rules of the type
	system of resource-aware SILL}.

\begin{figure}[t]
	\begin{small}
		\judgment{$\Phi; \Delta; q \vdash P :: (c: A)$}
		\begin{mathpar}
			\inferrule 
			{\Phi, x: b; \Delta, c: A; p \vdash P_{x} :: (d: D)} {\Phi; \Delta, c: b
				\land A; p \vdash x \leftarrow \m{recv} \; c; P_{x} :: (d: D)} {\land} L
			\and
			\inferrule 
			{\Phi; \Delta; p \vdash P :: (x: A) \\ \Phi \vdash e: b} {\Phi; \Delta; p
				\vdash \m{send} \; c \; e; P :: (c: b \land A)} {\land} R
			\and
			\inferrule
			{\Phi; \Delta, c_1: A_2; p \vdash P :: (d: D)} {\Phi; \Delta, c_2: A_1,
				c_1: A_1 \multimap A_2; p \vdash \\\\ \m{send} \; c_1 \; c_2; P :: (d: D)}
			{\multimap} L
			\and
			\inferrule
			{\Phi; \Delta, x: A_1; p \vdash P_{x} :: (c: A_2)} {\Phi; \Delta; p \vdash
				(x \leftarrow \m{recv} \; c; P_{x}) :: (c: A_1 \multimap A_2)} {\multimap}
			R
			\and
			\inferrule
			{\Phi; \Delta, x: A_1, c: A_2; p \vdash P_{x} :: (d: D)} {\Phi; \Delta, c:
				A_1 \otimes A_2; p \vdash \\\\ x \leftarrow \m{recv} \; c; P_{x} :: (d:
				D)} {\otimes} L
			\and
			\inferrule
			{\Phi; \Delta; p \vdash P :: (c_1: A_2)} {\Phi; \Delta, c_2: A_1; p \vdash
				\\\\ \m{send} \; c_1 \; c_2; P :: (c_1: A_1 \otimes A_2)} {\otimes} R
			\and
			\inferrule
			{\Phi; \Delta, c: A_{k}; p \vdash P :: (d: D)} {\Phi; \Delta, c: \&
				\{\ell_{i}: A_{i}\}; p \vdash \\\\ c. \ell_{k}; P : (d: D)} {\&} L
			\and
			\inferrule
			{\forall i. \Phi; \Delta; p \vdash P_{i} :: (c: A_{i})} {\Phi; \Delta; p
				\vdash \m{case} \; c \; \{ \overline{\ell_{i} \hookrightarrow P_{i}} \} ::
				(c: \& \{ \overline{\ell_{i} : A_{i}} \})} {\&} R
			\and
			\inferrule
			{\forall i. \Phi; \Delta, c: A_{i}; p \vdash P_{i} :: (d: D)} {\Phi;
				\Delta, c: \oplus \{ \overline{\ell_{i}: A_{i}} \}; p \vdash \\\\ \m{case}
				\; c \; \{ \overline{\ell_{i} \hookrightarrow P_{i}} \} :: (d: D)}
			{\oplus} L
			\and
			\inferrule
			{\Phi; \Delta; p \vdash P :: (c: A_{k})} {\Phi; \Delta; q \vdash
				(c.\ell_{k}; P) :: (c: \oplus \{ \overline{\ell_{i} : A_{i}} \})} {\oplus}
			R
			\and
			\inferrule
			{ } {\Phi; c_2: A; 0 \vdash c_1 \leftarrow c_2 :: (c_1: A)}
			\m{fwd}
			\and
			\inferrule
			{\Phi; \Delta; p \vdash P :: (d: D)} {\Phi; \Delta, c: \mathbf{1}; p
				\vdash \m{wait} \; c; P :: (d: D)} \mathbf{1} L
			\and
			\inferrule
			{ } {\Phi; \cdot; 0 \vdash \m{close} \; c :: (c: \mathbf{1})} \mathbf{1} R
			\and
			\inferrule
			{\Phi; \Delta, c: A; p \vdash P :: (d: D)} {\Phi; \Delta, c:
				\triangleleft^{q} A; p + q \vdash \m{pay} \; c \; \{q\}; P :: (d: D)}
			{\triangleleft} L
			\and
			\inferrule
			{\Phi; \Delta; p + q \vdash P :: (c: A)} {\Phi; \Delta; p \vdash \m{get}
				\; c \; \{q\}; P :: (c: \triangleleft^{q} A)} {\triangleleft} R
			\and
			\inferrule
			{\Phi; \Phi; \Delta, c: A; p + q \vdash P :: (d: D)} {\Phi; \Phi; \Delta,
				c: \triangleright^{q} A; p \vdash \m{get} \; c \; \{q\}; P :: (d: D)}
			{\triangleright} L
			\and
			\inferrule
			{\Phi; \Delta; p \vdash P :: (c: A)} {\Phi; \Delta; p + q \vdash \m{pay}
				\; c \; \{q\}; P :: (c: \triangleright^{q} A)} {\triangleright} R
		\end{mathpar}
	\end{small}
	\caption{Remaining rules of the type system of resource-aware SILL. $\Phi
			\vdash e : \tau$ in the rules $\m{spawn}$ and ${\supset} L$ is a typing
		judgment for the functional layer.}
	\label{fig:remaining rules of the type system of resource-aware SILL}
\end{figure}


\section{Session Skeletons}

\subsection{Checking Finiteness of Input Portions}
\label{appendix:Checking Finiteness of Input Portions}

The SMT-LIB2 encoding of \cref{eq:constraint on function f that encodes an
	infinite stream of Booleans} is displayed below.
\begin{small}
\begin{Verbatim}
(set-logic UFLIA)
(declare-fun f (Int) Bool)
(assert (forall ((x Int)) 
    (and (= (f (* 2 x)) true) 
         (= (f (+ (* 2 x) 1)) false))))
(check-sat)
\end{Verbatim}
\end{small}

How do we check the finiteness of an input portion?
Our approach is to count the number of input actions at the type level and check
if the number is infinite.
Given a skeleton $K$, the judgment
\begin{equation}
	\label{eq:judgment for the number of input actions in a skeleton}
	\Gamma \vdash \countinput{K}{n}
\end{equation}
states that $K$ contains at most $n$ many input actions at the type level, where
$n \in \bbN \cup \{\infty\}$.
A context $\Gamma$ maps type variables (i.e., $X$ in $\mu X. K_{X}$) to their
associated numbers of input actions.
Here, input actions are defined from the viewpoint of the channel's provider.
We also have the dual judgment $\Gamma \vdash \countoutput{K}{n}$ stating that
$K$ contains $n$ output actions at the type level.

The judgment \labelcref{eq:judgment for the number of input actions in a
	skeleton} is defined in \cref{fig:number of input actions in a skeleton}.
The dual judgment is defined similarly, so we omit its definition.
To compute the number of input and output actions, we construct a template
derivation tree according to this inference system.
In the tree, we use variables $n$'s to record the number of input/output
actions, and collect constraints on them.
Finally, we solve them by an LP solver.

\begin{figure}[t]
	\begin{small}
		\judgment{$\Gamma \vdash \countinput{K}{n}$}
		\begin{mathpar}
			\inferrule
			{ } {\Gamma \vdash \countinput{\mathbf{1}}{0}}
			\and
			\inferrule
			{\Gamma \vdash \countinput{K_2}{n}} {\Gamma \vdash \countinput{H \supset
					K}{n + 1}}
			\and
			\inferrule
			{\Gamma \vdash \countinput{K}{n}} {\Gamma \vdash \countinput{b \land
					K}{n}}
			\and
			\inferrule
			{\Gamma \vdash \countoutput{K_1}{n_1} \\ \Gamma \vdash
				\countinput{K_2}{n_2}} {\Gamma \vdash \countinput{K_1 \multimap K_2}{n_1
					+ n_2 + 1}}
			\and
			\inferrule
			{\Gamma \vdash \countinput{K_1}{n_1} \\ \Gamma \vdash
				\countinput{K_2}{n_2}} {\Gamma \vdash \countinput{K_1 \otimes K_2}{n_1 +
					n_2}}
			\and
			\inferrule
			{\forall i \in N. \Gamma \vdash \countinput{K_i}{n_i}} {\Gamma \vdash
				\countinput{\&_{x} \{\ell_{i}: K_{i} \mid i \in N \}}{1 + \max_{i \in N}
					\{n_i\}}}
			\and
			\inferrule
			{\Gamma \vdash \countinput{K}{n}} {\Gamma \vdash
				\countinput{\triangleleft^{q} K}{n + 1}}
			\and
			\inferrule
			{\forall i \in N. \Gamma \vdash \countinput{K_i}{n_i}} {\Gamma \vdash
				\countinput{\oplus \{\ell_{i}: K_{i} \mid i \in N \}}{\max_{i \in N}
					\{n_i\}}}
			\and
			\inferrule
			{\Gamma \vdash \countinput{K}{n}} {\Gamma \vdash
				\countinput{\triangleright^{q} K}{n}}
		\end{mathpar}
	\end{small}
	\caption{Number of input actions in a skeleton.}
	\label{fig:number of input actions in a skeleton}
\end{figure}

\subsection{Extraction of Cost Bounds}
\label{appendix:Extraction of Cost Bounds}

A type-level path of a session type (or a skeleton) is a path on the session
type where all choices of labels (including external and internal choices) are
resolved.

\begin{defi}[Type-level paths]
	\label{def:definition of type-level paths}
	Given a session type/skeleton $A$, its set of type-level paths is
	\begin{small}
		\begin{align*}
			\tpath{A_1 \multimap A_2}                   & \coloneq \{p_1 \multimap p_2 \mid p_i \in \tpath{A_i} \}               \\
			\tpath{A_1 \otimes A_2}                     & \coloneq \{p_1 \otimes p_2 \mid p_i \in \tpath{A_i} \}                 \\
			\tpath{\& \{\ell_i: A_i \mid i \in N\}}     & \coloneq \{\& \{\ell_k: p_k\} \mid k \in N, p_k \in \tpath{A_k} \}     \\
			\tpath{\oplus \{\ell_i: A_i \mid i \in N\}} & \coloneq \{\oplus \{\ell_k: p_k\} \mid k \in N, p_k \in \tpath{A_k} \} \\
			\tpath{\mathbf{1}}                          & \coloneq \{\mathbf{1}\}.
		\end{align*}
	\end{small}
	The sets of type-level paths for other type constructors can be defined
	straightforwardly.
\end{defi}

The issue illustrated by \cref{eq:code of P in the example illustrating the
	difficulty of extracting a cost bound} arises from the interactive nature of
resource-aware SILL.
A process in concurrent programming receives incoming messages (i.e., input) and
sends outgoing messages (i.e., output) that may depend on the previous input.
Afterwards, the process repeats the process of receiving input and sending
output.
In this way, input and output are intertwined in SILL.
As a consequence of the interdependence between input and output across
different channels, not all combinations of type-level paths on different
channels are feasible.
Furthermore, if some type-level paths have higher cost bounds than others but
are infeasible, it becomes challenging to figure out which type-level paths to
explore during worst-case input generation.


\section{Worst-Case Input Generation}

\subsection{Problem Statement}
\label{appendix:problem statement}

An input to a SILL program is a collection of predicates $\msg{\cdot, \cdot}$
from the external world.
To formalize inputs in SILL, we fix a naming scheme for channels in the
rewriting system.
As of now, whenever a fresh channel name is needed, the rewriting system does
not specify the fresh name.
Consequently, when we formally define an input as a collection of predicates
$\msg{c, M}$, it is unclear what precisely $c$ should be.
Thus, to formally define inputs, we must determine fresh channel names
deterministically.
Although it is possible to devise a desirable naming scheme for channels, due to
its complexity, we omit its formal definition.
Throughout this section, we adopt such a naming scheme.

Each input is a multiset of predicates $\msg{c, M}$.
Given a channel $c: K$ (where $K$ is a skeleton), let $\inputs{K, c}$ denote the
set of possible inputs from the perspective of $c$'s provider.
It is defined in \cref{fig:set of all possible inputs of a skeleton}.
The dual $\outputs{K, c}$ is the set of all outputs for $c$'s provider.
Since $\outputs{\cdot, \cdot}$ is defined similarly, we omit its definition.

\begin{defi}[Set of possible inputs]
	\label{def:set of possible inputs}
	Consider a network of processes with well-typed external channels $c_1: K_1,
		\ldots, c_{m}: K_{m}, c_{m+1}: K_{m+1}, \ldots, c_n: K_n$.
	Here, $c_1, \ldots, c_{m}$ are provided by processes inside the network, and
	$c_{m+1}, \ldots, c_{n}$ are provided by the external world.
	$K_1, \ldots, K_m$ are skeletons compatible with their original session types.
	The set of all possible inputs to this network is given by
	\begin{equation}
		\interp{K_1, \ldots, K_n} \coloneq \prod_{1 \leq i \leq m} \inputs{K_i, c_i} \times \prod_{m < i \leq n}\outputs{K_{i}, c_{i}}.
	\end{equation}
\end{defi}

\begin{figure}[t]
	\begin{small}
		\begin{align*}
			\inputs{H \supset K, c}                               & \coloneq \{t \cup \{ \msg{c, \m{send} \; c \; v} \} \mid v \in \interp{H}, t \in \inputs{K, c'} \}                           \\
			\inputs{b \land K, c}                                 & \coloneq \inputs{K, c}                                                                                                     \\
			\inputs{K_1 \multimap K_2, c}                         & \coloneq \{t_1 \cup t_2 \cup \{\msg{c', \m{send} \; c \; d; c' \leftarrow c}\} \mid                                        \\
			                                                      & \qquad c', d \text{ are fresh}, t_1 \in \inputs{K_2, c'}, t_2 \in \outputs{K_1, d} \}                                      \\
			\inputs{K_1 \otimes K_2, c}                           & \coloneq \{t_1 \cup t_2 \mid c', d \text{ are fresh}, t_1 \in \inputs{K_2, c'}, t_2 \in \inputs{K_1, d} \}                 \\
			\inputs{\& \{\ell_{i} : K_{i} \mid i \in N \}, c}     & \coloneq \{t \cup \{\msg{c', c. \ell_{k}; c' \leftarrow c}\} \mid j \in N, c'\text{ is fresh}, t \in \inputs{K_{j}, c'} \} \\
			\inputs{\oplus \{\ell_{i} : K_{i} \mid i \in N \}, c} & \coloneq \{t \mid j \in N, c'\text{ is fresh}, t \in \inputs{K_{j}, c'} \}                                                 \\
			\inputs{\mathbf{1}, c}                                & \coloneq \{\emptyset\}
		\end{align*}
	\end{small}
	\caption{Set of possible inputs of a skeleton. The dual $\outputs{K, c}$ is
		the set of all outputs for $c$'s provider.}
	\label{fig:set of all possible inputs of a skeleton}
\end{figure}

\begin{defi}[Worst-case inputs]
	\label{def:worst-case inputs}
	A worst-case input $(t_1, \ldots, t_n) \in \interp{K_1, \ldots, K_n}$ is such
	that, if we add $\bigcup_{1 \leq i \leq n} t_{i}$ to the initial configuration
	of the network and run it, we obtain the highest high-water-mark cost of all
	possible inputs from $\interp{K_1, \ldots, K_n}$.
\end{defi}

\subsection{Checking Tightness of Cost Bounds}
\label{appendix:checking tightness of cost bounds}

\begin{prop}[Checking depletion of \textcolor{red}{red} potential]
	\label{prop:checking red potential}
	Suppose red potential is tracked as described above without encountering the
	following issues:
	\begin{itemize}
		\item Red potential remains in some process at the end of symbolic
		      execution;
		\item Red potential is thrown away or flows back to the external world.
	\end{itemize}
	Then the red potential supplied to the program is completely consumed.
\end{prop}

\begin{proof}
	Whenever potential flows from one process to another, we assume that red
	potential (if there is any) in the sender is split evenly, and half of it is
	transferred.
	Throughout the symbolic execution, the flag $r$ is $\true$ if and only if the
	current process contains red potential.
	This invariant is proved by case analysis.
	Firstly, when potential is supplied by the external world, the Boolean flag of
	the recipient is set to $\false$, and it is consistent with the invariant.
	Secondly, when potential is transferred, the Boolean flag is correctly updated
	in a way that preserves the invariant.
	Lastly, red potential is never discarded.
	Thanks to the Boolean flag's invariant, when the symbolic execution
	successfully finishes, red potential should be completely gone because
	processes terminate only when $r = \true$ (i.e., red potential is absent).
	Therefore, if the symbolic execution successfully terminates, red potential
	will have been consumed completely.

	We split red potential evenly when potential is transferred to another
	process.
	Even if red potential is split differently, it does not affect our conclusion.
	For example, suppose we split red potential such that it stays in the sender
	or it all goes to the recipient.
	Then the set of processes with red potential is a subset of what we will have
	when red potential is split evenly.
	When the symbolic execution terminates, if red potential is divided evenly,
	the set of processes with red potential is empty.
	Hence, even if we change the way red potential is split, by the end of the
	symbolic execution, red potential must be completely gone.

	In conclusion, regardless of how we split red potential, when symbolic
	execution successfully terminates, all red potential is gone.
\end{proof}

\begin{citedthm}[\ref{theorem:checking tightness of cost bounds}]
	If red and blue potential is tracked without encountering issues, then the
	cost bound is tight.
\end{citedthm}
\begin{proof}
	A cost bound is equal to the amount of (red) potential supplied by the
	external world to a SILL program.
	It follows from \cref{prop:checking red potential} that red potential is
	entirely consumed if the symbolic execution successfully terminates.
	Also, because the symbolic execution properly tracks blue potential, there
	should be no path from unconsumed blue potential to red potential in
	\cref{fig:chain of the happened-before relation}.

	To show the tightness of a cost bound, it suffices to explicitly construct a
	schedule of processes whose high-water mark is equal to the cost bound.
	Firstly, suppose we only have one process $P$.
	If red potential is consumed completely, then when the last red potential is
	consumed, the high-water mark of $P$ is equal to the total red potential
	supplied by the external world.
	This can be seen from \cref{fig:tight cost bounds} (b).

	Next, consider a non-trivial case where we have two processes: $P_1$ and
	$P_2$.
	Let $t_1$ be the moment in $P_1$'s timeline when it completely consumes all
	potential in \cref{fig:chain of the happened-before relation}, including red
	potential.
	Define $t_2$ similarly for $P_2$.
	We can then run $P_1$ and $P_2$ until they stop exactly at $t_1$ and $t_2$,
	respectively.
	For example, if $P_1$ sends a message and $P_2$ receives it before $t_2$, then
	$t_1$ must happen after the event of $P_1$ sending the message.
	Because potential is completely consumed at $t_2$, any potential generated
	before $t_2$, including potential on $P_1$ right before sending the message,
	must be gone before $P_1$ reaches $t_1$.
	Furthermore, the high-water mark when $t_1$ and $t_2$ are reached is equal to
	the total red potential supplied by the external world.
	No potential is captured by messages in transit; otherwise, it would
	contradict the assumption that all potential in \cref{fig:chain of the
		happened-before relation} is entirely consumed by the time $t_1$ and $t_2$.
	Thus, all potential in \cref{fig:chain of the happened-before relation},
	including red potential, should have been consumed by $P_1$ and $P_2$ when
	they reach $t_1$ and $t_2$.
	Therefore, the net cost at this point is equal to the total red potential
	supplied by the external world.

	This reasoning can be generalized to more than two processes.
\end{proof}

\subsection{Soundness and Relative Completeness}
\label{appendix:soundness and relative completeness}

\begin{defi}[Similarity between configurations]
	\label{def:similarity between configurations}
	Fix $S$ to be a solution to a final path constraint generated once the
	symbolic execution terminates.
	Let $t$ be an input (i.e., a multiset of predicates $\msg{c, M}$) induced by
	the solution $S$ to the final path constraint.
	Consider some configuration $C_{\text{sym}}$ during the symbolic execution and
	a configuration $C_{\text{cost}}$ for the cost semantics.
	The similarity relation $C_{\text{sym}} \sim C_{\text{cost}}$ holds if and
	only if there is an injection from the multiset $t \cup C_{\text{sym}}$ of
	predicates to the multiset $C_{\text{cost}}$ of predicates such that each pair
	in the injection satisfies the similarity relation (\cref{def:similarity
		between predicates}).
\end{defi}

All proofs related to the soundness and relative completeness of worst-case
input generation are presented in this section.

\begin{citedthm}[\ref{prop:simulation for soundness}]
	Suppose we are given three configurations: $C_{1, \text{sym}}$, $C_{2,
				\text{sym}}$, and $C_{1, \text{cost}}$.
	The first two configurations are used in the symbolic execution, and the last
	one is used in the cost semantics.
	These configurations satisfy two conditions: (i) $C_{1, \text{sym}}$
	transitions to $C_{2, \text{sym}}$ in one step of the symbolic execution and
	(ii) $C_{1, \text{sym}} \sim C_{1, \text{cost}}$ holds.
	Then there exists a configuration $C_{2, \text{cost}}$ of the cost semantics
	such that the following diagram commutes:
	\begin{equation}
		\xymatrix{
		C_{1, \text{sym}} \ar[r]_{w} & C_{2, \text{sym}} \\
		C_{1, \text{cost}} \ar@{~}[u] \ar[r]_{w}^>{\leq 1} & C_{2, \text{cost}} \ar@{~}[u]
		}
	\end{equation}
	In this diagram, $C_{1, \text{sym}} \xrightarrow[w]{} C_{2, \text{sym}}$ means
	$\costboundconfig{C_{1, \text{sym}}} - \costboundconfig{C_{2, \text{sym}}} =
		w$, where $\costboundconfig{\cdot}$ denotes a cost bound of a configuration
	(see \cref{def:cost bounds of configurations}).
	Likewise, $C_{1, \text{cost}} \xrightarrow[w]{} C_{2, \text{cost}}$ means
	$C_{1, \text{cost}}$ transitions to $C_{2, \text{cost}}$ such that the net
	cost increases by $w$.
	The arrow $\rightarrow^{\leq 1}$ means the number of steps is either zero or
	one.
\end{citedthm}
\begin{proof}
	By case analysis on the rewriting rules of the symbolic execution.
	Strictly speaking, in the rules for termination and forwarding in symbolic
	execution, potential may be discarded.
	Therefore, the above commutative diagram is not quite correct: after one
	transition step in both the symbolic execution and cost semantics, the
	potential may decrease by $w$, while the net cost stays the same.
	However, this can be fixed by saving all potential, including the one that is
	actually discarded before termination and forwarding in the symbolic
	execution.
\end{proof}

\begin{prop}[Simulation for completeness]
	\label{prop:simulation for completeness}
	Suppose we are given three configurations: $C_{1, \text{cost}}$, $C_{2,
				\text{cost}}$, and $C_{1, \text{sym}}$.
	The first two configurations are used in the cost semantics, and the last one
	is used in the symbolic execution.
	These configurations satisfy two conditions: (i) $C_{1, \text{cost}}$
	transitions to $C_{2, \text{cost}}$ in one step of cost semantics and (ii)
	$C_{1, \text{sym}} \sim C_{1, \text{cost}}$ holds.
	Then there exists a configuration $C_{2, \text{sym}}$ of the symbolic
	execution such that the following diagram commutes:
	\begin{equation}
		\xymatrix{
		C_{1, \text{sym}} \ar[r]_{w}^>{\geq 1} & C_{2, \text{sym}} \\
		C_{1, \text{cost}} \ar@{~}[u] \ar[r]_{w} & C_{2, \text{cost}} \ar@{~}[u]
		}
	\end{equation}
	In this diagram, $C_{1, \text{cost}} \xrightarrow[w]{} C_{2, \text{cost}}$
	means $C_{1, \text{cost}}$ transitions to $C_{2, \text{cost}}$ such that the
	net cost increases by $w$.
	The arrow $\rightarrow^{\geq 1}$ means the number of steps is at least one.
	Likewise, $C_{1, \text{sym}} \xrightarrow[w]{} C_{2, \text{sym}}$ means
	$\costboundconfig{C_{1, \text{sym}}} - \costboundconfig{C_{2, \text{sym}}} =
		w$.
\end{prop}
\begin{proof}
	By case analysis on the rewriting rules of the cost semantics (\cref{fig:some
		rules of the cost semantics of SILL}).
	In the symbolic execution, when processes terminate or forward, the processes
	are sometimes allowed to throw away potential.
	As a result, this breaks the above commutative diagram because potential may
	decreases while the net cost stays the same.
	To work around this issue, as done in the proof of \Cref{prop:simulation for
		soundness}, we save potential somewhere instead of throwing it away.
\end{proof}

\subsection{Symbolic Execution}
\label{appendix:symbolic execution}

Some of the key rules for the symbolic execution are already presented in
\cref{fig:key rules in the symbolic execution for the process layer}.
The remaining key rules are given in \cref{fig:remaining key rules in the
	symbolic execution for the process layer part 1,fig:remaining key rules in the
	symbolic execution for the process layer part 2}.
\Cref{fig:key rules in the symbolic execution for the process
	layer,fig:remaining key rules in the symbolic execution for the process layer
	part 1,fig:remaining key rules in the symbolic execution for the process layer
	part 2} cover half of all rules.
The other half is just the dual of the three figures in this article; hence, it
is omitted.

In the symbolic execution for the process layer, we need to transfer potential.
Hence, we augment the grammar of message $M$ (\cref{appendix:cost semantics and
	the type system}) as follows:
\begin{align*}
	M \Coloneq {} & \cdots \mid \m{pay} \; c \; \{q\}; c \leftarrow c' \mid \m{pay} \; c \; \{q\}; c' \leftarrow c.
\end{align*}
Here, $q \in \bbQ_{> 0}$ denotes the quantity of potential to be transferred.

Skeleton variables of functional types are added to path constraints during the
functional layer's symbolic execution.
For instance, suppose a process's code contains $\m{if} \; b \; \m{then} \; e_1
	\; \m{else} \; e_2$.
During the symbolic execution, if we choose to explore the first branch, we add
$b$ to a path constraint.
As the symbolic execution for the functional layer is already presented in a
prior work~\cite{Wang2019}, this article omits it.

\begin{figure}[t]
	\begin{small}
		\begin{mathpar}
			\inferrule
			{\proc{\Delta_1, \Delta_2; p + q \vdash (c \leftarrow e \leftarrow
					\overline{c_i}; Q_{c}) :: (d: D), \phi_2, (\ids{s}, \ids{p})} \\ e
				\Downarrow \tuple{\phi_1, x \leftarrow P_{x, \overline{x_i}}
					\leftarrow \overline{x_i}} \\ c' \text{ is fresh}} {\proc{\Delta_1; p
					\vdash P_{c',\overline{ c_{i}}} :: (c' : A), \phi_1, (\ids{s},
					\cternary{p = 0}{\emptyset}{\ids{p}}) } \\ \proc{\Delta_2, c': A; q
					\vdash Q_{c'} :: (d: D), \phi_2, (\ids{s}, \cternary{q =
						0}{\emptyset}{\ids{p}} ) }}
			\m{spawn}
			\and
			\inferrule
			{\msg{c', \m{send} \; c \; v; c' \leftarrow c, \phi_1, (\ids{s, 1},
					\emptyset)} \\ \proc{\Delta; p \vdash x \leftarrow \m{recv} \; c;
					P_{x} :: (c: b \supset A), \phi_2, (\ids{s, 2}, \ids{p, 2})}}
			{\proc{\Delta; p \vdash P_{v} [c' / c] :: (c: A), \phi_1 \land \phi_2,
					(\ids{s, 1} \cup \ids{s, 2}, \ids{p, 2}) }} {\supset}
			R_{\text{internal}}
			\and
			\inferrule
			{\proc{\Delta; p \vdash x \leftarrow \m{recv} \; c; P_{x} :: (c: H \supset
					K), \phi, (\ids{s}, \ids{p})}} {\proc{\Delta; p \vdash P_{H} [c' / c]
					:: (c: K), \phi, (\ids{s}, \ids{p})}} {\supset} R_{\text{external}}
			\and
			\inferrule
			{\proc{\Delta, c_1: A_1 \multimap A_2, c_2: A_1; p \vdash \m{send} \; c_1
					\; c_2; P :: (d: D), \phi, (\ids{s}, \ids{p}) } \\\\ c_1' \text{ is
					fresh}} {\proc{\Delta, c_1': A_2; p \vdash P [c_1' / c_1] :: (d: D),
					\phi, (\ids{s}, \ids{p})} \\\\ \msg{c_1', \m{send} \; c_1 \; c_2; c_1'
					\leftarrow c_1, \top, (\ids{s}, \emptyset)}} {\multimap} S
			\and
			\inferrule
			{\msg{c_1', \m{send} \; c_1 \; c_2; c_1' \leftarrow c_1, \top, (\ids{s,
						1}, \emptyset)} \\ \proc{\Delta; p \vdash x \leftarrow \m{recv} \;
					c_1; P_{x} :: (c_1 :: A_1 \multimap A_2), \phi, (\ids{s, 2}, \ids{p,
						2})}} {\proc{\Delta; p \vdash P_{c_2} [c_1' / c_1] :: (c_1': A_2),
					\phi, (\ids{s, 1} \cup \ids{s, 2}, \ids{p, 2})}} {\multimap}
			R_{\text{internal}}
			\and
			\inferrule
			{\proc{\Delta; p \vdash x \leftarrow \m{recv} \; c_1; P_{x} :: (c_1 :: A_1
					\multimap A_2), \phi, (\ids{s}, \ids{p})} \\ c_1' \text{ is fresh}}
			{\proc{\Delta; p \vdash P_{c_2} [c_1' / c_1] :: (c_1': A_2), \phi,
					(\ids{s}, \ids{p})}} {\multimap} R_{\text{external}}
			\and
			\inferrule
			{\proc{\Delta, c: \& \{\ell_i : A_i \mid i \in N\}; p \vdash c.\ell_{k}; P
					:: (d: D), \phi, (\ids{s}, \ids{p})} \\ c' \text{ is fresh}}
			{\proc{\Delta, c': A_k; p \vdash P [c' / c] :: (d: D), \phi, (\ids{s},
					\ids{p})} \\ \msg{c', c.\ell_{k}; c' \leftarrow c, \top, (\ids{s},
					\emptyset)}} {\&} S
			\and
			\inferrule
			{\msg{c', c. \ell_{k}; c' \leftarrow c, \top, (\ids{s, 1}, \emptyset)} \\
				\proc{\Delta; p \vdash \m{case} \; c \; \{\ell_{i} \hookrightarrow P_{i}
					\mid i \in N\} :: (c: \& \{\ell_i: A_i \mid i \in N \}), \phi,
					(\ids{s, 2}, \ids{p, 2})}} {\proc{\Delta; p \vdash P_{k} [c' / c] ::
					(c': A_k), \phi, (\ids{s, 1} \cup \ids{s, 2}, \ids{p, 2})}} {\&}
			R_{\text{internal}}
		\end{mathpar}
	\end{small}
	\caption{Remaining key rules in the symbolic execution for the process layer
		(part 1). The judgment $e \Downarrow \tuple{\phi, v}$ means $e$ evaluates to
		a (symbolic) value $v$ with a path constraint $\phi$.}
	\label{fig:remaining key rules in the symbolic execution for the process layer part 1}
\end{figure}

\begin{figure}[t]
	\begin{small}
		\begin{mathpar}
			\inferrule
			{\proc{c_2: A; p \vdash c_1 \leftarrow c_2 :: (c_1: A), \phi, (\ids{s},
					\ids{p})} \\ \textcolor{red}{\text{red}} \notin \ids{p}} {\msg{c_1,
					c_1 \leftarrow c_2, \phi}} \m{fwd}_{s}
			\and
			\inferrule
			{\proc{\Delta; p \vdash P :: (c_2: A), \phi_1, (\ids{s, 1}, \ids{p, 1})}
				\\ \msg{c_1, c_1 \leftarrow c_2, \phi_2, (\ids{s, 2}, \emptyset)}}
			{\proc{\Delta; p \vdash P [c_1 / c_2] :: (c_1: A), \phi_1 \land \phi_2,
					(\ids{s, 1} \cup \ids{s, 2}, \ids{p, 1})}} \m{fwd}_{r}^{+}
			\and
			\inferrule
			{\msg{c_1, c_1 \leftarrow c_2, \phi_1, (\ids{s, 1}, \emptyset)} \\
				\proc{\Delta, c_1: A; p \vdash P :: (d: D), \phi_2, (\ids{s, 2}, \ids{p,
						2})}} {\proc{\Delta, c_2: A; p \vdash P :: (d: D), \phi_1 \land
					\phi_2, (\ids{s, 1} \cup \ids{s, 2}, \ids{p, 2})}} \m{fwd}_{r}^{-}
			\and
			\inferrule
			{\msg{c, \m{close} \; c, \phi_1, (\ids{s, 1}, \emptyset)} \\ \proc{\Delta,
					c: \mathbf{1}; p \vdash \m{wait} \; c; P :: (d: D), \phi_2, (\ids{s,
						2}, \ids{p, 2})}} {\proc{\Delta; p \vdash P :: (d: D), \phi_1 \land
					\phi_2, (\ids{s, 1} \cup \ids{s, 2}, \ids{p})}} \mathbf{1}
			R_{\text{internal}}
			\and
			\inferrule
			{\proc{\Delta, c: \mathbf{1}; p \vdash \m{wait} \; c; P :: (d: D), \phi_2,
					(\ids{s}, \ids{p})}} {\proc{\Delta; p \vdash P :: (d: D), \phi_1 \land
					\phi_2, (\ids{s}, \ids{p})}} \mathbf{1} R_{\text{external}}
			\and
			\inferrule
			{\proc{\Delta, c: \triangleleft^{q} A; p+q \vdash \m{pay} \; c \; \{q\}; P
					:: (d: B), \phi, (\ids{s}, \ids{p})} \\ c' \text{ is fresh}}
			{\proc{\Delta, c': A; p \vdash P [c' / c] :: (d: B), \phi, (\ids{s},
					\cternary{p = 0}{\emptyset}{\ids{p}} ) } \\ \msg{c', \m{pay} \; c \;
					\{p\}; c' \leftarrow c, (\ids{s}, \ids{p})}} {\triangleleft}
			L_{\text{internal}}
			\and
			\inferrule
			{\proc{\Delta; q \vdash \m{get} \; c \; \{p\}; P :: (c: \triangleleft^{p}
					A), \phi, (\ids{s, 1}, \ids{p, 1})} \\ \msg{c', \m{pay} \; c \; \{p\};
					c \leftarrow c', (\ids{s, 2}, \ids{p, 2})}} {\proc{\Delta; p+q \vdash
					P [c' / c] :: (c': A), \phi, (\ids{s, 1} \cup \ids{s, 2}, \ids{p, 1}
					\cup \ids{p, 2}) }} {\triangleleft} R_{\text{internal}}
		\end{mathpar}
	\end{small}
	\caption{Remaining key rules in the symbolic execution for the process layer (part 2).}
	\label{fig:remaining key rules in the symbolic execution for the process layer part 2}
\end{figure}


\section{Case Study: Server and Browsers}
\label{appendix:case study}

The implementation of a server with independent sessions (\cref{sec:independent
	sessions in the case study}) is
\begin{equation}
	d_1 \leftarrow P \leftarrow \cdot; \m{send} \; c \; d_1; d_2 \leftarrow P \leftarrow \cdot; \m{send} \; c \; d_2; \m{close} \; c,
\end{equation}
where the process $d \leftarrow P \leftarrow \cdot$ is implemented as
\begin{equation}
	\label{eq:implementation of the server where channels are created by the server}
	\begin{split}
		P \coloneq {} & x \leftarrow \m{recv} \; d; \m{tick} \; 1; \m{send} \; d \; \tuple{x+1, y};             \\
		              & \m{case} \; d \; \{ \m{ack} \hookrightarrow z \leftarrow \m{recv} \; d;                 \\
		              & \qquad \qquad \qquad \;\m{if} \; (z = 1 + y) \; \m{then}                                \\
		              & \qquad \qquad \qquad \quad d.\m{success}; \m{tick} \; 1; \m{tick} \; -2; \m{close} \; d \\
		              & \qquad \qquad \qquad \; \m{else}                                                        \\
		              & \qquad \qquad \qquad \quad d.\m{failure}; \m{close} \; d,                               \\
		              & \qquad \quad \; \m{timeout} \hookrightarrow \m{tick} \; -1; \m{close} \; d \}.
	\end{split}
\end{equation}

Integers $x$ and $y$ are called sequence numbers and are stored in the server.
This is why we have $\m{tick} \; 1$.
Additionally, once the handshake is completed successfully, we run $\m{tick} \;
	1$ because we assume that one memory cell is required for the subsequent
communication phase.
Lastly, before a channel is closed, we free up all memory.

With a sequential scheduler, $P$ is implemented as
\begin{equation}
	\label{eq:implementation of the server where channels are created by the external world}
	\begin{split}
		P \coloneq {} & x_1 \leftarrow \m{recv} \; d_1; \m{tick} \; 1; \m{send} \; d_1 \; \tuple{x_1 + 1, y_1};                                                      \\
		              & \m{case} \; d_1 \; \{\m{ack} \hookrightarrow \ldots; x_2 \leftarrow \m{recv} \; d_2; \m{tick} \; 1; \m{send} \; d_2 \; \tuple{x_2 + 1, y_2}; \\
		              & \qquad \qquad \qquad \quad \m{case} \; d_2 \; \{ \m{ack} \hookrightarrow \ldots, \m{timeout} \hookrightarrow \ldots \},                      \\
		              & \qquad \qquad \m{timeout} \hookrightarrow \ldots \}.
	\end{split}
\end{equation}

\Cref{sec:coordinating sessions with a scheduler in the case study} studies a
web server capable of scheduling sessions.
With a round-robin scheduler, the server is
\begin{equation}
	d_1 \leftarrow \m{recv} \; c; d_2 \leftarrow \m{recv} \; c; c \leftarrow P \leftarrow d_1, d_2,
\end{equation}
where the process $c \leftarrow P \leftarrow d_1, d_2$ is implemented as
\begin{equation}
	\begin{split}
		P \coloneq {} & x_1 \leftarrow \m{recv} \; d_1; x_2 \leftarrow \m{recv} \; d_2;                                                                                \\
		              & \m{tick} \; 1; \m{tick} \; 1; \m{send} \; d_1 \; \tuple{x_1 + 1, y_1}; \m{send} \; d_2 \; \tuple{x_2 + 1, y_2}                                 \\
		              & \m{case} \; d_1 \; \{\m{ack} \hookrightarrow \m{case} \; d_2 \; \{ \ldots \}, \m{timeout} \hookrightarrow \m{case} \; d_2 \; \{ \ldots \} \}.
	\end{split}
\end{equation}

With the round-robin scheduler, resource-annotated $A$ is
\begin{equation}
	A_{\text{anno}} \coloneq \integer \land \textcolor{red}{\triangleright^{1}} (\integer \times \integer) \supset \oplus \{ \m{ack}: \integer \land \& \{\m{success}: \textcolor{red}{\triangleright^{1}} \mathbf{1}, \m{failure}: \mathbf{1} \}, \m{timeout}: \mathbf{1} \}.
\end{equation}
Hence, the overall cost bound according to this resource-annotated session type
is 4 (i.e., 2 for each instance of $A$).
This is identical to the cost bound from \cref{sec:independent sessions in the
	case study}, and the worst-case input generation algorithm generates the same
worst-case input.

\end{document}